\documentclass[12pt]{article}%

\usepackage{tikz}  
\usepackage{pifont}%
\usepackage[symbol*]{footmisc}

\usepackage{hyperref}%
\hypersetup{%
   breaklinks,%
   ocgcolorlinks,
   colorlinks=true,%
   linkcolor=[rgb]{0.45,0.0,0.0},%
   citecolor=[rgb]{0,0,0.45}
}
\usepackage{titlesec}
\titlelabel{\thetitle. }

\usepackage{ntheorem}
\usepackage{wide}%
\usepackage{color}%
\usepackage{xspace}%
\usepackage{euscript}%
\usepackage{amstext}
\usepackage{paralist}%
\usepackage{amsmath}%
\usepackage{amssymb}%
\usepackage{graphicx}
\usepackage{theorem}%
\usepackage{salgorithm}%

\theoremseparator{.}

\newcommand*\circled[1]{\footnotesize\tikz[baseline=(char.base)]{
    \node[shape=circle,draw,inner sep=0.2pt] (char) {#1};}}

\DefineFNsymbols*{stars}[text]{%
   {\ding{172}} %
   {\ding{173}} %
   {\ding{174}} %
   {\ding{175}} %
   {\ding{176}} %
   {\ding{177}} %
   {\ding{178}} %
   {\ding{179}} %
   {\ding{180}} %
   {{\protect\circled{{\tiny 10}}}}%
   {{\protect\circled{{\tiny 11}}}}%
   {{\protect\circled{{\tiny 12}}}}%
   {{\protect\circled{{\tiny 13}}}}%
   {{\protect\circled{{\tiny 14}}}}%
   {{\protect\circled{{\tiny 15}}}}%
   {{\protect\circled{{\tiny 16}}}}%
   {{\protect\circled{{\tiny 17}}}}%
   {{\protect\circled{{\tiny 18}}}}%
   {{\protect\circled{{\tiny 19}}}}%
   {{\protect\circled{{\tiny 20}}}}%
   {{\protect\circled{{\tiny 21}}}}%
   {{\protect\circled{{\tiny 22}}}}%
}%

\providecommand{\SarielWWWPapersAddr}
        {http://www.uiuc.edu/\string~sariel/papers}

\definecolor{blue25}{rgb}{0,0,0.25}
\newcommand{\emphic}[2]{%
     \textcolor{blue25}{%
         \textbf{\emph{#1}}}%
         \index{#2}}
\newcommand{\emphi}[1]{\emphic{#1}{#1}}

\newcommand{\obslab}[1]{\label{observation:#1}}
\newcommand{\obsref}[1]{Observation~\ref{observation:#1}}

\newcommand{\assumplab}[1]{\label{assumption:#1}}
\newcommand{\assumpref}[1]{Assumption~\ref{assumption:#1}}

\providecommand{\lemlab}[1]{\label{lemma:#1}}
\providecommand{\lemref}[1]{Lemma~\ref{lemma:#1}}

\newcommand{\figlab}[1]{\label{fig:#1}}
\newcommand{\figref}[1]{Figure~\ref{fig:#1}}

\newtheorem{theorem}{Theorem}[section] 
\newtheorem{lemma}[theorem]{Lemma}

\newtheorem{defn}[theorem]{Definition}
\newtheorem{observation}[theorem]{Observation}
\newtheorem{assumption}[theorem]{Assumption}
\newcommand{\Term}[1]{\textsf{#1}}
\newcommand{\TermI}[1]{\Term{#1}\index{#1@\Term{#1}}}
\newcommand{\BFS}{\textsf{BFS}\xspace}

\newcommand{\aftermathA}{\par\vspace{-\baselineskip}}
\theoremstyle{remark}{\theorembodyfont{\rm} \newtheorem{remark}[theorem]{Remark}}
\newenvironment{proof}{\trivlist\item[]\emph{Proof}:}%
                  {\unskip\nobreak\hskip 1em plus 1fil\nobreak%
                           \rule{2mm}{2mm}
                           \parfillskip=0pt%
                           \endtrivlist}
\newcommand{\thmlab}[1]{{\label{theo:#1}}}
\newcommand{\thmref}[1]{Theorem~\ref{theo:#1}}

\newcommand{\eqlab}[1]{\label{equation:#1}}
\newcommand{\Eqref}[1]{Eq.~(\ref{equation:#1})}
\newcommand{\Eqrefpage}[1]{Eq.~(\ref{equation:#1})%
   $_\text{p\pageref{equation:#1}}$}
\newcommand{\seclab}[1]{{\label{section:#1}}}
\newcommand{\secref}[1]{Section~\ref{section:#1}}

\newcommand{\deflab}[1]{\label{defn:#1}}
\newcommand{\defref}[1]{Definition~\ref{defn:#1}}
\newcommand{\defrefpage}[1]{Definition~\ref{defn:#1}%
   $_\text{p\pageref{defn:#1}}$}

\newcommand{\MakeVBig}{\rule[-.2cm]{0cm}{0.62cm}}
\newcommand{\MakeBig}{\rule[-.2cm]{0cm}{0.4cm}}
\newcommand{\MakeSBig}{\rule[0.0cm]{0.0cm}{0.35cm}} 
\newcommand{\brc}[1]{\left\{ {#1} \right\}}
\newcommand{\sep}[1]{\,\left|\, {#1} \MakeBig\right.}
\newcommand{\pth}[2][\!]{#1\left({#2}\right)}
\newcommand{\pbrc}[2][\!\!]{#1\left[ {#2} \MakeBig \right]}

\newcommand{\ceiling}[1]{\left \lceil {#1} \right \rceil}
\newcommand{\norm}[1]{\left\lVert {#1} \right \rVert}

\newcommand{\normY}[2]{\left\lVert {#1} - {#2} \right \rVert}

\newcommand{\normP}[1]{\norm{#1}_{\oplus}}

\newcommand{\distP}[2]{\normP{{#1}- {#2}}}
\newcommand{\distPk}[3]{\mathsf{d}_{#3}\pth{#2,#1}}
\newcommand{\distNN}[2]{\mathsf{d}\pth{#1,#2}}
\newcommand{\minDistPk}[2]{r_{\mathrm{opt}}\pth{{#1},{#2}}}
\newcommand{\radius}[1]{\mathrm{radius}({#1})}

\newcommand{\abs}[1]{\left | {#1} \right |}
\newcommand{\floor}[1]{\left\lfloor {#1} \right\rfloor}
\newcommand{\nfrac}[2]{{#1}/{#2}}
\newcommand{\cardin}[1]{\left\lvert {#1} \right\rvert}
\newcommand{\order}[1]{O\pth{#1}}

\newcommand{\ordereq}[1]{\Theta \left ( {#1} \right )}

\newcommand{\eps}{{\varepsilon}}%
\newcommand{\divides}{|}

\newcommand{\DDM}[2]{F}
\newcommand{\aDDM}[2]{{F}_1}
\newcommand{\ADDM}[2]{F_2}
\newcommand{\IdxSet}{\mathcal{I}}

\newcommand{\wt}[1]{w_{#1}}
\newcommand{\wtX}[1]{w\pth{#1}}
\newcommand{\FSF}{\mathcal{F}_{\mathsf{sg}}}
\newcommand{\SarielThanks}[1]{\thanks{Department of Computer
      Science; 
      University of Illinois; 
      201 N. Goodwin Avenue;
      Urbana, IL, 61801, USA;
      {\tt sariel\atgen{}illinois.edu}; {\tt
         \url{http://www.illinois.edu/\string~sariel/}.} #1}}
\newcommand{\NirmanThanks}[1]{\thanks{Department of Computer
      Science; 
      University of Illinois; 
      201 N. Goodwin Avenue;
      Urbana, IL, 61801, USA;
      {\tt \si{nkumar5}\atgen{}illinois.edu}; {\tt
         \url{http://www.illinois.edu/\string~\si{nkumar5}/}.} #1}}
\newcommand{\atgen}{\symbol{'100}}

\providecommand{\si}[1]{#1}

\newcommand{\WSPD}{\TermI{WSPD}\xspace}
\newcommand{\ANN}{\TermI{ANN}\xspace}
\newcommand{\NNTerm}{\TermI{NN}\xspace}

\newcommand{\AVD}{\TermI{AVD}\xspace}
\newcommand{\etal}{\textit{et~al.}\xspace}
\renewcommand{\Re}{{\rm I\!\hspace{-0.025em} R}}
\newcommand{\diameter}[1]{\mathsf{diam}\pth{ {#1} }}
\newcommand{\PntSet}{\mathsf{P}}
\newcommand{\PntSetQ}{\mathsf{Q}}
\newcommand{\query}{\mathtt{q}}

\newcommand{\pnt} {\mathsf{p}}
\newcommand{\pntA}{\mathsf{u}}
\newcommand{\pntB}{\mathsf{v}}
\newcommand{\pntC}{\mathsf{s}}

\newcommand{\distS}[2]{\mathsf{d}\pth{#1,#2}}
\newcommand{\distSP}[2]{\mathsf{d}_{\oplus}\pth{#1,#2}}

\newcommand{\dist}[2]{\norm{{#1}- {#2}}}

\newcommand{\distY}[2]{\normY{#1}{#2}}

\newcommand{\distXS}[2]{\mathsf{d}_\oplus\pth{#1,#2}}
\newcommand{\ballA}{\mathsf{b}}

\newcommand{\BallSet}{\mathcal{B}}
\newcommand{\ctrA}{\mathsf{c}}
\newcommand{\radA}{\mathsf{r}}

\newcommand{\BallA}{\mathsf{B}}
\newcommand{\CellSetA}{\mathcal{X}}

\newcommand{\CellSetAVD}{\mathcal{W}}
\newcommand{\CellSetC}{\mathcal{S}}

\newcommand{\cellA}{\mathsf{\Box}}

\newcommand{\Scube}{U}
\newcommand{\UnitCube}{[0,1]^d}

\newcommand{\mapped}[1]{#1'}

\newcommand{\repres}[1]{#1_{\mathsf{rep}}}
\newcommand{\knnrepX}[1]{\mathrm{nn}_k\pth{#1}}
\newcommand{\DFS}{\Algorithm{DFS}\xspace}

\newcommand{\algQC}{\Algorithm{\si{QuorumCluster}}\xspace}
\renewcommand{\th}{th\xspace}
\newcommand{\Otilde}{\widetilde{O}}

\newcommand{\NNk}[3]{\mathsf{nn}_{#1}\pth{#2, #3}}%

\newcommand{\Grid}{\mathsf{G}\index{grid}}
\newcommand{\ball}[2]{\mathsf{ball}\pth{#1,#2}}
\newcommand{\gridSet}[2]{\boxplus\pth{#1, #2}}

\newcommand{\num}{\alpha}
\newcommand{\ds}{\displaystyle}
\newcommand{\QTree}{\EuScript{T}\index{quadtree}}
\newcommand{\repX}[1]{\mathrm{nn'}\pth{#1}}
\newcommand{\pntrepX}[1]{\mathbf{p}\pth{#1}}
\newcommand{\remove}[1]{}
\newcommand{\qtreenode}{\nu}

\newcommand{\rvector}{\mathbf{b}}

\newcommand{\adknn}[1]{\mathrm{\beta}_k\pth{#1}}
\newcommand{\approxkd}{R}

\newcommand{\lo}{\mathsf{l}}%
\newcommand{\hi}{\mathsf{h}}%
\newcommand{\levelX}[1]{\mathrm{l{v}l}\pth{#1}}

\newcommand{\ceil}[1]{\left\lceil {#1} \right\rceil}

\newcommand{\constA}{\zeta_1}
\newcommand{\constB}{c}
\newcommand{\constC}{\zeta_2}
\newcommand{\constD}{\zeta_3}
\newcommand{\ts}{\hspace{0.6pt}}

\newlength{\savedparindent}
\newcommand{\SaveIndent}{\setlength{\savedparindent}{\parindent}}
\newcommand{\RestoreIndent}{\setlength{\parindent}{\savedparindent}}

\newcommand{\numA}{z}
\newcommand{\DS}{\mathcal{D}}%
\newcommand{\Array}{\mathcal{X}}

\newlength{\ppicwd}
\newcommand{\padbox}[1]{%
   \settowidth{\ppicwd}{#1}%
   \begin{minipage}{1.01\ppicwd}%
       \smallskip%
       {#1}%
       \smallskip%
   \end{minipage}
}

\newcommand{\wx}{\tau}

\newcommand{\hdim}{d}

\newcommand{\Pk}[2][\!]{\PntSet_{\leq k}\pth[#1]{#2}}
\newcommand{\PkExt}[4][\!]{{#2}_{\leq #3}\pth[#1]{#4}}
\newcommand{\Family}{\mathcal{H}}%

\newcommand{\dk}[2][\!]{\mathsf{d}_{k}\pth[#1]{#2}}
\newcommand{\BadProb}{\varphi}
\newcommand{\prob}{\rho}
\newcommand{\RangeSpace}{\mathsf{S}}
\newcommand{\GroundSet}{\textsf{X}}%
\newcommand{\pbrcS}[1]{\left[ {#1} \right]}
\newcommand{\Dim}{\delta}%

\newcommand{\DistX}[2]{d_{#1}\pth{#2}}
\newcommand{\MeasureChar}{\overline{m}}

\newcommand{\Measure}[1]{\MeasureChar\pth{#1}}
\newcommand{\sMeasureX}[2]{\overline{s}_{#2}\pth{#1}}
\newcommand{\pbrcx}[1]{\left[ {#1} \right]}
\newcommand{\Prob}[1]{\mathop{\mathbf{Pr}}\!\pbrcx{#1}}

\newcommand{\X}{\EuScript{X}}%
\newcommand{\R}{\EuScript{R}}%
\newcommand{\range}{\mathsf{r}}%
\newcommand{\VC}{\Term{VC}\xspace}
\newcommand{\RSample}{\mathsf{R}}

\newcommand{\Error}{\EuScript{E}}
\providecommand{\Merigot}{M{\' e}rigot}%

\begin{document}

\title{Down the Rabbit Hole: Robust Proximity Search and Density
   Estimation in Sublinear Space%
   \footnote{%
      %
      Work on this paper was partially supported by NSF AF awards
      CCF-0915984 and CCF-1217462. %
      A preliminary version of this paper appeared in FOCS 2012
      \cite{hk-drhrp-12}.}%
}

\author{%
   Sariel Har-Peled%
   \SarielThanks{}%
   \and%
   Nirman Kumar%
   \NirmanThanks{}}

\date{\today}

\maketitle

\setfnsymbol{stars}

\begin{abstract}
    For a set of $n$ points in $\Re^d$, and parameters $k$ and $\eps$,
    we present a data structure that answers $(1+\eps,k)$-\ANN queries
    in logarithmic time.  Surprisingly, the space used by the
    data-structure is $\Otilde (n /k)$; that is, the space used is
    sublinear in the input size if $k$ is sufficiently large. Our
    approach provides a novel way to summarize geometric data, such
    that meaningful proximity queries on the data can be carried out
    using this sketch. Using this, we provide a sublinear space
    data-structure that can estimate the density of a point set under
    various measures, including:
    \begin{inparaenum}[(i)]
        \item sum of distances of $k$ closest points to the query
        point, and
        \item sum of squared distances of $k$ closest points to the
        query point.
    \end{inparaenum}
    Our approach generalizes to other distance based estimation of
    densities of similar flavor.

    We also study the problem of approximating some of these
    quantities when using sampling. In particular, we show that a
    sample of size $\Otilde (n /k)$ is sufficient, in some restricted
    cases, to estimate the above quantities. Remarkably, the sample
    size has only linear dependency on the dimension.
\end{abstract}

\section{Introduction}

Given a set $\PntSet$ of $n$ points in $\Re^d$, the \emphi{nearest
   neighbor} problem is to construct a data structure, such that for
any \emph{query} point $\query$ it (quickly) finds the closest point
to $\query$ in $\PntSet$.  This is an important and fundamental
problem in Computer Science \cite{sdi-nnmlv-06, c-tpfgn-08,
   ai-nohaa-08, c-nnsms-06}. Applications of nearest neighbor search
include pattern recognition \cite{fh-dandc-49, ch-nnpc-67},
self-organizing maps \cite{k-som-01}, information retrieval
\cite{swy-vsmai-75}, vector compression \cite{gg-vqsc-91},
computational statistics \cite{dw-nnmd-82}, clustering
\cite{dhs-pc-01}, data mining, learning, and many others.  If one is
interested in guaranteed performance and near linear space, there is
no known way to solve this problem efficiently (i.e., logarithmic
query time) for dimension $d > 2$.

A commonly used approach for this problem is to use Voronoi
diagrams. The \emphi{Voronoi diagram} of $\PntSet$ is the
decomposition of $\Re^d$ into interior disjoint closed cells, so that
for each cell $C$ there is a unique single point $\pnt \in \PntSet$
such that for any point $\query \in \mathrm{int}\pth{C}$ the
nearest-neighbor of $\query$ in $\PntSet$ is $\pnt$.  Thus, one can
compute the nearest neighbor of $\query$ by a point location query in
the collection of Voronoi cells.  In the plane, this approach leads to
$O(\log n)$ query time, using $O(n)$ space, and preprocessing time
$O(n \log n)$.  However, in higher dimensions, this solution leads to
algorithms with \emph{exponential} dependency on the dimension. The
complexity of a Voronoi diagram of $n$ points in $\Re^d$ is
$\ordereq{n^{\ceiling{d/2}}}$ in the worst case.  By requiring
slightly more space, Clarkson \cite{c-racpq-88} showed a
data-structure with query time $\order{\log n }$, and
$\order{n^{\ceiling{d/2} + \delta}}$ space, where $\delta > 0$ is a
prespecified constant (the $O(\cdot)$ notation here hides constants
that are exponential in the dimension). One can tradeoff the space
used and the query time \cite{am-rsps-93}. Meiser \cite{m-plah-93}
provided a data-structure with query time $\order{d^5 \log n }$ (which
has polynomial dependency on the dimension), where the space used is
$\order{ n^{ d + \delta}}$.  Therefore, even for moderate dimension,
the exact nearest neighbor data-structure uses an exorbitant amount of
storage.  It is believed that there is no efficient solution for the
nearest neighbor problem when the dimension is sufficiently large
\cite{mp-p-69}; this difficulty has been referred to as the ``curse of
dimensionality''.

\paragraph{Approximate Nearest Neighbor (\ANN).}

In light of the above, major effort has been devoted to develop approximation
algorithms for nearest neighbor search \cite{amnsw-oaann-98,
   im-anntr-98, kor-esann-00, sdi-nnmlv-06, c-tpfgn-08, ai-nohaa-08,
   \si{c-nnsms-06}, \si{him-anntr-12}}. In the 
\emphi{$(1+\eps)$-approximate nearest neighbor} problem (the \ANN problem), 
one is additionally given an  approximation parameter $\eps > 0$ and one 
is required to find a point $\pntA \in \PntSet$ such that 
$\dist{\query}{\pntA} \leq (1+\eps) \distS{\query}{\PntSet}$. In $d$ dimensional 
Euclidean space, one can answer \ANN queries, in $O(\log n + 1/\eps^{d-1})$ time
using linear space \cite{amnsw-oaann-98, h-gaa-11}. Because of the
$1/\eps^{d-1}$ in the query time, this approach is only efficient in
low dimensions. Interestingly, for this data-structure, the
approximation parameter $\eps$ need not be specified during the
construction, and one can provide it during the query.  An alternative
approach is to use Approximate Voronoi Diagrams (\AVD), introduced by
Har-Peled \cite{h-rvdnl-01}, which is a partition of space into regions
of low total complexity, with a representative point for each region,
that is an \ANN for any point in the region. In particular, Har-Peled
showed that there is such a decomposition of size $O\pth{(n
   /\eps^d)\log^2 n}$, see also \cite{him-anntr-12}. 
This allows \ANN queries to be answered in $O(
\log n)$ time.  Arya and Malamatos \cite{am-lsavd-02} showed how to
build \AVD{}s of linear complexity (i.e., $O(n/\eps^d)$). Their
construction uses \WSPD (Well Separated Pair Decomposition)
\cite{ck-dmpsa-95}. Further tradeoffs between query time and space usage for
\AVD{}s were studied by Arya \etal \cite{amm-sttan-09}.

\paragraph{$k$-nearest neighbor.}
A more general problem is the $k$-nearest neighbors problem where one
is interested in finding the $k$ points in $\PntSet$ nearest to the
query point $\query$. This is widely used in pattern recognition,
where the majority label is used to label the query point. In this paper, 
we are interested in the more restricted problem of approximating the
distance to the $k$\th nearest neighbor and finding a data point
achieving the approximation. We call this problem the 
\emphi{$(1+\eps,k)$-approximate nearest neighbor} ($(1+\eps,k)$-\ANN) 
problem. This problem is widely used for density
estimation in statistics, with $k \approx \sqrt{n}$
\cite{s-desda-86}. It is also used in meshing (with $k=2d$), or to compute
the local feature size of a point set in $\Re^d$ \cite{r-draqt-95}.
The problem also has applications in non-linear dimensionality
reduction; finding low dimensional structures in data -- more
specifically low dimensional submanifolds embedded in Euclidean
spaces. Algorithms like ISOMAP, LLE, Hessian-LLE, SDE and others,
use the $k$-nearest neighbor as a subroutine \cite{t-mmpo-98,
   bslt-gagem-00, ms-trn-94, ws-ulms-04}.

\paragraph{Density estimation.}
Given distributions $\mu_1, \ldots, \mu_k$ defined over $\Re^d$, and a
query point $\query$, we want to compute the \emph{a
   posteriori} probabilities of $\query$ being generated by one of
these distributions. This approach is used in unsupervised learning as
a way to classify a new point. Naturally, in most cases, the
distributions are given implicitly; that is, one is given a large
number of points sampled from each distribution. So, let $\mu$ be such
a distribution, and $\PntSet$ be a set of $n$ samples. To estimate the
density of $\mu$ at $\query$, a standard Monte Carlo technique is to
consider a ball $\BallA$ centered at $\query$, and count the number of
points of $\PntSet$ inside $\BallA$.  Specifically, one possible
approach that is used in practice \cite{dhs-pc-01}, is to find the
smallest ball centered at $\query$ that contains $k$ points of
$\PntSet$ and use this to estimate the density of $\mu$.  The
right value of $k$ has to be chosen carefully -- if it is too small,
then the estimate is unstable (unreliable), and if it is too large, it
either requires the set $\PntSet$ to be larger, or the estimate is too
``smoothed'' out to be useful (values of $k$ that are used in practice
are $\Otilde(\sqrt{n})$), see Duda \etal \cite{dhs-pc-01} for more
details. To do such density estimation, one needs to be able to answer,
approximate or exact, $k$-nearest neighbor queries.

Sometimes one is interested not only in the radius of this ball
centered at the query point, but also in the distribution of the
points inside this ball. The average distance of a point inside the
ball to its center, can be estimated by the sum of distances of the sample
points inside the ball to the center. Similarly, the variance of this
distance can be estimated by the sum of squared distances of the
sample points inside the ball to the center of the ball.  As mentioned, 
density estimation is used in manifold learning and surface
reconstruction. For example, Guibas \etal \cite{gmm-wkd-11} recently
used a similar density estimate to do manifold reconstruction.

\paragraph{Answering exact $k$-nearest neighbor queries.}

Given a point set $\PntSet \subseteq \Re^d$, computing the partition
of space into regions, such that the $k$ nearest neighbors do not
change, is equivalent to computing the \emphi{$k$\th order Voronoi
   diagram}.  Via standard lifting, this is equivalent to computing
the first $k$ levels in an arrangement of hyperplanes in $\Re^{d+1}$
\cite{a-vdsfg-91}. More precisely, if we are interested in the
$k$\th-nearest neighbor, we need to compute the $(k-1)$-level in this
arrangement.

The complexity of the $(\leq k)$ levels of a hyperplane arrangement in
$\Re^{d+1}$ is $\Theta( n^{\floor{(d+1)/2}} (k+1)^{\ceiling{(d+1)/2}}
)$ \cite{cs-arscg-89}. The exact complexity of the $k$\th-level is not
completely understood and achieving tight bounds on its complexity is
one of the long-standing open problems in discrete geometry
\cite{m-ldg-02}. In particular, via an averaging argument, in the
worst case, the complexity of the $k$\th-level is $\Omega
\pth{n^{\floor{(d+1)/2} } (k+1)^{\ceiling{(d+1)/2} -1}}$. As such, the
complexity of $k$\th-order Voronoi diagram is $\Omega(n k)$ in two
dimensions, and $\Omega(n^2 k)$ in three dimensions.

Thus, to provide a data-structure for answering $k$-nearest neighbor
queries exactly and quickly (i.e., logarithmic query time) in $\Re^d$,
requires computing the $k$-level of an arrangement of hyperplanes in
$\Re^{d+1}$. The space complexity of this structure is prohibitive
even in two dimensions (this also effects the preprocessing time).
Furthermore, naturally, the complexity of this structure increases as
$k$ increases.  On the other end of the spectrum one can use
partition-trees and parametric search to answer such queries using
linear space and query time (roughly) $O\pth{n^{1-1/(d+1)}}$
\cite{m-ept-92, c-opt-10}. One can get intermediate results using
standard space/time tradeoffs \cite{ae-rsir-98}.

\paragraph{Known results on approximate $k$-order Voronoi %
   diagram.}
Similar to \AVD, one can define a \AVD for the $k$-nearest neighbor.
The case $k=1$ is the regular approximate Voronoi diagram
\cite{h-rvdnl-01, am-lsavd-02, amm-sttan-09}. The case $k=n$ is the
furthest neighbor Voronoi diagram. It is not hard to see that it has a
constant size approximation (see \cite{h-caspm-99}, although it was
probably known before). Our results (see below) can be interpreted as
bridging between these two extremes.

\paragraph{Quorum clustering.}
Carmi \etal \cite{cdhks-gqsa-05} describe how to compute efficiently a
partition of the given point set $\PntSet$ into clusters of $k$ points each,
such that the clusters are compact. Specifically, this quorum clustering
computes the smallest ball containing $k$ points, removes
this cluster, and repeats, see \secref{quorum} for more details. Carmi
\etal \cite{cdhks-gqsa-05} also describe a data-structure that can
approximate the smallest cluster. The space usage of their data structure is
$\Otilde(n /k)$, but it cannot be directly used for our
purposes. Furthermore, their data-structure is for two dimensions and
it cannot be extended to higher dimensions, as it uses additive
Voronoi diagrams (which have high complexity in higher dimensions).

\section*{Our results.}

We first show, in \secref{const}, how to build a data-structure
that answers $(15,k)$-\ANN queries in time $O( \log n)$, 
where the input is a
set of $n$ points in $\Re^d$. Surprisingly, the space used by this
data-structure is $O(n/k)$.  This result is surprising as the
space usage \emphi{decreases} with $k$. This is in sharp contrast to
behavior in the exact version of the $k$\th-order Voronoi diagram
(where the complexity increases with $k$). Furthermore, for
super-constant $k$ the space used by this data-structure is
sublinear. For example, in some applications the value of $k$ used is
$\Omega\pth{ \sqrt{n}}$, and the space used in this case is a tiny
fraction of the input size.  This is a general reduction showing that
such queries can be reduced to proximity search in an appropriate
product space over $n/k$ points computed carefully.

In \secref{sec:avd}, we show how to construct an \emph{approximate}
$k$-order Voronoi diagram using space $O(\eps^{-d-1} n/k)$ (here
$\eps>0$ is an approximation quality parameter specified in
advance). Using this data-structure one can answer
$(1+\eps,k)$-\ANN queries in $O(\log n)$
time. See \thmref{ann:main} for the exact result.

\paragraph{General density queries.}
We show in \secref{applications}, as an application of our
data-structure, how to answer more robust queries.  For
example, one can approximate (in roughly the same time and space as
above) the sum of distances, or squared distances, from a query point
to its $k$ nearest neighbors. This is useful in approximating density
measures \cite{dhs-pc-01}. Surprisingly, our data-structure can be
used to estimate the sum of any function $f(\cdot)$ defined over the
$k$ nearest neighbors, that depends only on the distance of these
points from the query point. Informally, we require that $f(\cdot)$ is
monotonically increasing with distance, and it is (roughly) not
super-polynomial. For example, for any constant $p > 0$, our
data-structure requires sublinear space (i.e., $\Otilde \pth{
   n/k}^{\MakeBig}$), and given a query point $\query$, it can
$(1+\eps)$-approximate the quantity $\sum_{\pntA \in X}
\dist{\pntA}{\query}^p$, where $X$ is the set of $k$ nearest points 
in $\PntSet$ to $\query$. The query time is logarithmic.

To facilitate this, in a side result, that might be of independent
interest, we show how to perform point-location queries in $I$
compressed quadtrees of total size $m$ simultaneously in $O( \log m +
I)$ time (instead of the naive $O( I \log m)$ query time), without
asymptotically increasing the space needed.

\paragraph{If $k$ is specified with the query.}
In \secref{qtree:algo}, given a set $\PntSet$ of $n$ points in
$\Re^d$, we show how to build a data-structure, in $O(n \log n)$ time
and using $O(n)$ space, such that given a query point and parameters
$k$ and $\eps$, the data-structure can answer $(1+\eps,k)$-\ANN 
queries in $O( \log n + 1/\eps^{d-1})$
time. Unlike previous results, this is the first data-structure where
\emph{both} $k$ and $\eps$ are specified during the query
time. The data-structure of Arya \etal \cite{amm-sttas-05}
required knowing $\eps$ in advance. Using standard techniques
\cite{amnsw-oaann-98} to implement it, should lead to a simple and
practical algorithm for this problem.

\paragraph{If $k$ is not important.}
Note, that our main result can not be done using sampling. Indeed,
sampling is indifferent to the kind of geometric error we care
about. Nevertheless, a related question is how to answer a
$(1+\eps,k)$-\ANN query if one is allowed to also approximate $k$.
Inherently, this is a different question that is, at least
conceptually, easier. Indeed, the problem boils down to using sampling
carefully, and loses much of its geometric flavor. We show to solve
this variant (this seems to be new) in \secref{sampling}. Furthermore,
we study what kind of density functions can be approximated by such an
approach. Interestingly, the sample size needed to provide good
density estimates is of size $\Otilde(n/k)$ (which is sublinear in
$n$), and surprisingly, has only linear dependency on the
dimension. This compares favorably with our main result, where the
space requirement is exponential in the dimension.

\paragraph{Techniques used.}
We use quorum clustering as a starting point in our solution. In
particular, we show how it can be used to get a constant
factor approximation to the approximate $k$-nearest neighbor distance
using sublinear space.  Next, we extend this construction and combine
it with ideas used in the computation of approximate Voronoi
diagrams. This results in an algorithm for computing approximate
$k$-nearest neighbor Voronoi diagram.  To extend this data-structure
to answer general density queries, as described above, requires a
way to estimate the function $f(\cdot)$ for relatively few values
(instead of $k$ values) when answering a query. We use a coreset
construction to find out which values need to be approximated.
Overall, our work combines several known techniques in a non-trivial
fashion, together with some new ideas, to get our new results.

For the sampling results, of \secref{sampling}, we need to use some
sampling bounds that are not widely known in Computational Geometry.

\paragraph{Paper organization.}
In \secref{prelim} we formally define the problem and introduce some
basic tools, including quorum clustering, which is a key insight into
the problem at hand.  The ``generic'' constant factor algorithm is
described in \secref{const}.  We describe the construction of the
approximate $k$-order Voronoi diagram in \secref{sec:avd}.  In
\secref{applications} we describe how to construct a data-structure to
answer density queries of various types. In \secref{qtree:algo} we
present the data-structure for answering $k$-nearest neighbor queries
that does not require knowing $k$ and $\eps$ in advance.  The
approximation via sampling is presented in \secref{sampling}.
We conclude in \secref{conclusions}.

\section{Preliminaries}
\seclab{prelim}

\subsection{Problem definition}

Given a set $\PntSet$ of $n$ points in $\Re^d$ and a number $k$, $1 \leq k
\leq n$, consider a point $\query$ and order the points of $\PntSet$
by their distance from $\query$; that is,
\begin{align*}
    \dist{\query}{\pntA_1} \leq \dist{\query}{\pntA_2} \leq \dots \leq
    \dist{\query}{\pntA_n},
\end{align*}
where $\PntSet = \brc{\pntA_1, \pntA_2, \dots, \pntA_n}$. The point
$\pntA_k = \NNk{k}{\query}{\PntSet}$ is the \emphi{$k$\th-nearest
   neighbor} of $\query$ and $\distPk{\PntSet}{\query}{k} =
\dist{\query}{\pntA_k}$ is the \emphi{$k$\th-nearest neighbor
   distance}. The nearest neighbor distance (i.e., $k=1$) is
$\distS{\query}{\PntSet} = \min_{\pntA \in \PntSet}
\dist{\query}{\pntA}$.  The global minimum of
$\distPk{\PntSet}{\query}{k}$, denoted by $\minDistPk{\PntSet}{k} =
\min_{\query \in \Re^d} \distPk{\PntSet}{\query}{k}$, is the radius of
the smallest ball containing $k$ points of $\PntSet$. 

\begin{observation}%
    \obslab{1:Lipschitz}%
    For any $\pnt, \pntA \in \Re^d$, $k$ and a set $\PntSet \subseteq
    \Re^d$, we have that $\distPk{\PntSet}{\pntA}{k} \leq
    \distPk{\PntSet}{\pnt}{k} + \dist{\pnt}{\pntA}$.
\end{observation}
Namely, the function $\distPk{\PntSet}{\query}{k}$ is
$1$-Lipschitz.
The problem at hand is to preprocess $\PntSet$ such that given a query
point $\query$ one can compute $\pntA_k$ quickly.  The standard
\emphi{nearest neighbor problem} is this problem for $k = 1$. In the
\emphi{$(1+\eps,k)$-approximate nearest neighbor} ($(1+\eps,k)$-\ANN) problem, given
$\query$, $k$ and $\eps>0$, one wants to find a point $\pntA \in
\PntSet$, such that $(1-\eps) \dist{\query}{\pntA_k} \leq
\dist{\query}{\pntA} \leq (1+\eps)\dist{\query}{\pntA_k}$.

\subsection{Basic tools}

For a real positive number $\num$ and a point $\pnt = (\pnt_1, \ldots,
\pnt_d) \in \Re^d$, define $\Grid_\num(\pnt)$ to be the grid point
$\pth[]{\floor{\pnt_1/\num} \num, \ldots, \floor{\pnt_d/\num}
   \num}$. We call $\num$ the \emphi{width} or \emphi{sidelength} of
the \emphi{grid} $\Grid_\num$. Observe that the mapping $\Grid_\num$
partitions $\Re^d$ into cubic regions, which we call grid
\emphic{cells}{cell}.

\begin{defn}
    A cube is a \emphi{canonical cube} if it is contained inside the
    unit cube $[0,1]^d$, it is a cell in a grid $\Grid_r$, and $r$ is
    a power of two (i.e., it might correspond to a node in a quadtree
    having $[0,1]^d$ as its root cell).  We will refer to such a grid
    $\Grid_r$ as a \emphic{canonical grid}{canonical!grid}. Note, that
    all the cells corresponding to nodes of a compressed quadtree are
    canonical.
    \index{grid!canonical} %
    \deflab{canonical:grid}
\end{defn}

For a ball $\ballA$ of radius $r$, and a parameter $\psi$, let
$\gridSet{\ballA}{\psi}$ denote the set of all the canonical cells
intersecting $\ballA$, when considering the canonical grid with
sidelength $2^{\floor{\log_2 \psi }}$. Clearly,
$\cardin{\gridSet{\ballA}{\psi}} = O\pth{ (r/\psi)^d}$.

A ball $\ballA$ of radius $r$ in $\Re^d$, centered at a point $\pnt$,
can be interpreted as a point in $\Re^{d+1}$, denoted by
$\mapped{\ballA} = \pth[]{ \pnt, r}$. For a regular point $\pnt \in
\Re^d$, its corresponding image under this transformation is the
\emphi{mapped} point $\mapped{\pnt} = \pth[]{\pnt, 0 } \in \Re^{d+1}$.

Given point $\pntA = \pth{\pntA_1,\dots,\pntA_d} \in \Re^d$ we will
denote its Euclidean norm by $\norm{\pntA}$.  We will consider a point
$\pntA = \pth{\pntA_1, \pntA_2,\dots, \pntA_{d+1}} \in \Re^{d+1}$ to
be in the product metric of $\Re^d \times \Re$ and endowed with the
product metric norm
\begin{align*}
    \normP{\pntA} = \sqrt{\pntA_1^2 + \dots + \pntA_d^2} +
    \abs{\pntA_{d+1}}.
\end{align*}
It can be verified that the above defines a norm and the following
holds for it.
\begin{lemma}%
    \lemlab{p:norm}%
    For any $\pntA \in \Re^{d+1}$ we have $\norm{\pntA} \leq
    \normP{\pntA} \leq \sqrt{2} \norm{\pntA}$.
\end{lemma}
The distance of a point to a set under the $\normP{\cdot}$ norm is
denoted by $\distXS{\pntA}{\PntSet}$.

\begin{assumption}%
    \assumplab{k:div:n}%
    We assume that $k$ divides $n$; otherwise one can easily add fake
    points as necessary at infinity.
\end{assumption}

\begin{assumption}%
    \assumplab{points:in:cube}%
    We also assume that the point set $\PntSet$
    is contained in $[1/2,1/2 + 1/n]^d$, where
    $n=\cardin{\PntSet}$. This can be achieved by scaling and
    translation (which does not affect the distance ordering).
    Moreover, we assume the queries are restricted to the unit cube
    $\Scube = \UnitCube$.
\end{assumption}

\subsubsection{Quorum clustering}
\seclab{quorum}

\begin{figure*}[t]
    \centerline{
       \includegraphics{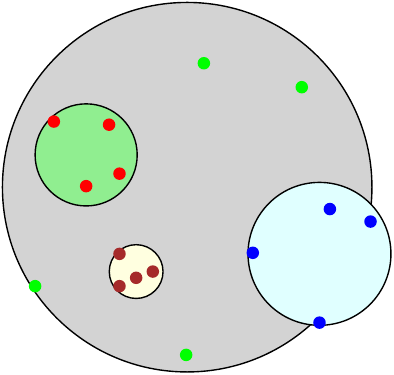}
    }
    \caption{Quorum clustering for $n = 16$ and $k = 4$.}
    \figlab{fig:quorum}
\end{figure*}

Given a set $\PntSet$ of $n$ points in $\Re^d$, and a number $k \geq
1$, where $k \divides n$, we start with the smallest ball $\ballA_1$
that contains $k$ points of $\PntSet$, that is $\radius{\ballA_1} =
\minDistPk{\PntSet}{k}$. Let $\PntSet_1 = \PntSet \cap
\ballA_1$. Continue on the set of points $\PntSet \setminus \PntSet_1$
by finding the smallest ball that contains $k$ points of $\PntSet
\setminus \PntSet_1$, and so on.  Let
$\ballA_1,\ballA_2,\dots,\ballA_{n/k}$ denote the set of balls
computed by this algorithm and let $\PntSet_i = \pth{\PntSet \setminus
   \pth{\PntSet_1 \cup \dots \cup \PntSet_{i-1}}} \cap \ballA_i$. See
\figref{fig:quorum} for an example. Let $\ctrA_i$ and $\radA_i$ denote
the center and radius respectively, of $\ballA_i$, for $i=1,\ldots,
n/k$.  A slight symbolic perturbation can guarantee that
\begin{inparaenum}[(i)]
    \item each ball $\ballA_i$ contains exactly $k$ points of
    $\PntSet$, and
    \item all the centers $\ctrA_1, \ctrA_2, \dots, \ctrA_k$, are
    distinct points.
\end{inparaenum}
Observe that $\radA_1 \leq \radA_2 \leq \dots \leq
\radA_{n/k} \leq \diameter{\PntSet}$.  Such a partition of $\PntSet$
into $n/k$ clusters is a \emphi{quorum clustering}. An
algorithm for computing it is provided in Carmi \etal
\cite{cdhks-gqsa-05}. We assume we have a black-box procedure
\algQC{}$(\PntSet,k)$ \cite{cdhks-gqsa-05} that computes an
\emphi{approximate} quorum clustering. It returns a list of balls,
$(\ctrA_1, \radA_1), \ldots, ({\ctrA_{n/k}, \radA_{n/k}})$.  The
algorithm of Carmi \etal \cite{cdhks-gqsa-05} computes such a sequence
of balls, where each ball is a $2$-approximation to the smallest
ball containing $k$ points of the remaining points. The following is
an improvement over the result of Carmi \etal \cite{cdhks-gqsa-05}.

\begin{lemma}%
    \lemlab{q:clustering}%
    Given a set $\PntSet$ of $n$ points in $\Re^d$ and parameter $k$,
    where $k \divides n$,
    one can compute, in $O(n \log n)$ time, a sequence of $n/k$ balls,
    such that, for all $i, 1 \leq i \leq n/k$, we have \smallskip
    \begin{compactenum}[\rm \quad(A)]
        \item For every ball $\ts\pth[]{\ctrA_i, \radA_i}$ there is an
        associated subset $\PntSet_i$ of $k$ points of
        $\PntSetQ_i = \PntSet\setminus \pth[]{ \PntSet_i \cup \ldots \cup
           \PntSet_{i-1}}$, that it covers.
        \item The ball $\ts \pth[]{\ctrA_i, \radA_i}$ is a
        $2$-approximation to the smallest ball covering $k$ points in
        $\PntSetQ_i$; that is, $\radA_i/2 \leq
        \minDistPk{\PntSetQ_i}{k} \leq \radA_i$.
    \end{compactenum}
\end{lemma}

\begin{proof}
    The guarantee of Carmi \etal is slightly worse -- their algorithm
    running time is $O(n \log^{d} n)$. They use a dynamic
    data-structure for answering $O(n)$ queries, that report how many
    points are inside a query canonical square. Since they use
    orthogonal range trees this requires $O( \log^d n)$ time per
    query. Instead, one can use dynamic quadtrees. More formally, we
    store the points using linear ordering \cite{h-gaa-11}, using any
    balanced data-structure. A query to decide the number of points
    inside a canonical node corresponds to an interval query (i.e.,
    reporting the number of elements that are inside a query interval),
    and can be performed in $O( \log n)$ time. Plugging this
    data-structure into the algorithm of Carmi \etal
    \cite{cdhks-gqsa-05} gives the desired result.
\end{proof}

\section{A $(15,k)$-\ANN in sublinear space}
\seclab{const}

\begin{lemma} %
    \lemlab{5:approx}%
    Let $\PntSet$ be a set of $n$ points in $\Re^d$, $k\geq 1$ be a
    number such that $k \divides n$, $\ts\pth[]{\ctrA_1,\radA_1},
    \ldots$, $\pth{\ctrA_{n/k},\radA_{n/k}}$, be the list of balls
    returned by \algQC{}$(\PntSet,k)$, and let $ x = \min_{i =
       1,\dots,n/k}$ $\pth{\dist{\query}{\ctrA_i} + \radA_i}$. We have
    that $x/5 \leq \distPk{\PntSet}{\query}{k} \leq x$.
\end{lemma}

\begin{proof}
    For any $i = 1, \ldots, n/k,$ we have $\ballA_i =
    \ball{\ctrA_i}{\radA_i} \subseteq
    \ball{\query}{\dist{\query}{\ctrA_i} + \radA_i}$.  Since
    $\cardin{\ballA_i \cap \PntSet} \geq k$, we have
    $\distPk{\PntSet}{\query}{k} \leq \dist{\query}{\ctrA_i} +
    \radA_i$.  As such, $\distPk{\PntSet}{\query}{k} \leq x = \ds
    \min_{i = 1, \dots, n/k} \pth{\dist{\query}{\ctrA_i} + \radA_i}$.
    
    For the other direction, let $i$ be the first index such that
    $\ball{\query}{\distPk{\PntSet}{\query}{k}}$ contains a point of
    $\PntSet_i$, where $\PntSet_i$ is the set of $k$ points of
    $\PntSet$ assigned to $\ballA_i$. Then, we have
    \begin{align*}
        \radA_i/2 %
        \leq %
        \minDistPk{\PntSetQ_i}{k}%
        \leq %
        \distPk{\PntSet}{\query}{k},%
    \end{align*}
    where $\PntSetQ_i = \PntSet \setminus (\PntSet_1 \cup \dots \cup
    \PntSet_{i-1})$, $\radA_i$ is a $2$-approximation to
    $\minDistPk{\PntSetQ_i}{k}$, and the last inequality follows
    as $X = \ball{\query}{\distPk{\PntSet}{\query}{k}} \cap \PntSet$
    is a set of size $k$ and $X \subseteq \PntSetQ_i$.
    Then,
    \begin{align*}
        \dist{\query}{\ctrA_i} - \radA_i%
        \leq%
        \distNN{\query}{\ball{\ctrA_i}{\radA_i}}%
        \leq%
        \distPk{\PntSet}{\query}{k} %
        ,
    \end{align*}
    as the distance from $\query$ to any $\pntA \in
    \ball{\ctrA_i}{\radA_i}$ satisfies $\dist{\query}{\pntA} \geq
    \dist{\query}{\ctrA_i} - \radA_i$ by the triangle inequality.
    Putting the above together, we get
    \begin{align*}
        x%
        = \min_{j=1,\dots,n/k} \pth{\dist{\query}{\ctrA_j} + \radA_j}%
        \leq %
        \dist{\query}{\ctrA_i} + \radA_i%
        = %
        \pth[]{\dist{\query}{\ctrA_i} - \radA_i} + 2\radA_i%
        \leq%
        5 \distPk{\PntSet}{\query}{k}.
    \end{align*}
    \aftermathA 
    \aftermathA 
\end{proof}

\begin{theorem}%
    \thmlab{const:main}%
    Given a set $\PntSet$ of $n$ points in $\Re^d$, and a number $k
    \geq 1$ such that $k \divides n$, one can build a data-structure,
    in $\order{n \log n}$ time, that uses $\order{n/k}$ space, such that
    given any query point $\query \in \Re^d$, one can compute, in
    $\order{\log \pth {n/k}}$ time, a $15$-approximation to
    $\distPk{\PntSet}{\query}{k}$.
\end{theorem}
\begin{proof}
    We invoke \algQC{}$(\PntSet,k)$ to compute the clusters
    $\pth[]{\ctrA_i, \radA_i}$, for $i=1,\ldots, n/k$. For $i=1,
    \ldots, n/k$, let $\mapped{\ballA_i} = \pth{\ctrA_i,\radA_i} \in
    \Re^{d+1}$. We preprocess the set $\mapped{\BallSet} =
    \brc{\mapped{\ballA_1}, \ldots, \mapped{\ballA_{n/k}}}$ for
    $2$-\ANN queries (in $\Re^{d+1}$ under the Euclidean norm). 
    The preprocessing time for the
    \ANN data structure is $\order{\pth{n/k} \log \pth{n/k}}$, 
    the space used is $\order{n/k}$ and the query time is 
    $\order{\log \pth{n/k}}$ \cite{h-gaa-11}.
    
    Given a query point $\query \in \Re^d$ the algorithm computes a
    $2$-ANN to $\mapped{\query} = \pth[]{\query, 0}$, denoted by
    $\mapped{ \ballA_j}$, and returns
    $\distP{\mapped{\query}}{\mapped{\ballA_j}}$ as the approximate
    distance.
    
    Observe that, for any $i$, we have $\dist{\mapped{\query}}{
       \mapped{\ballA_i}} \leq
    \distP{\mapped{\query}}{\mapped{\ballA_i}} \leq \sqrt{2}
    \dist{\mapped{\query}}{ \mapped{\ballA_i} }$ by \lemref{p:norm}.
    As such, the returned distance to $\mapped{\ballA_j}$ is a
    $2$-approximation to $
    \distS{\mapped{\query}}{\mapped{\BallSet}}$; that is,
    \begin{align*}
        {\distSP{\mapped{\query}}{\mapped{\BallSet}}}%
        \leq %
        \distP{\mapped{\query}}{\mapped{\ballA_j}}%
        \leq%
        \sqrt{2} \dist{\mapped{\query}}{\mapped{\ballA_j}}%
        \leq%
        2 \sqrt{2} \distS{\mapped{\query}}{\mapped{\BallSet}}%
        \leq%
        2 \sqrt{2} \distSP{\mapped{\query}}{\mapped{\BallSet}}.%
    \end{align*}
    By \lemref{5:approx}, $
    \distSP{\mapped{\query}}{\mapped{\BallSet}}/5 \leq
    \distPk{\PntSet}{\query}{k} \leq
    \distSP{\mapped{\query}}{\mapped{\BallSet}}$. Namely,
    \[
      \distP{\mapped{\query}}{\mapped{\ballA_j}}/(10\sqrt{2}) \leq
      \distPk{\PntSet}{\query}{k} \leq
      \distP{\mapped{\query}}{\mapped{\ballA_j}},
    \]
    implying the claim.
\end{proof}

\begin{remark}
    The algorithm of \thmref{const:main} works for any metric
    space. Given a set $\PntSet$ of $n$ points in a metric space, one
    can compute $n/k$ points in the product space induced by adding
    an extra coordinate, such that approximating the distance to the
    $k$\th nearest neighbor, is equivalent to answering \ANN queries
    on the reduced point set, in the product space.
\end{remark}

\section{Approximate Voronoi diagram for %
   $\distPk{\PntSet}{\query}{k}$}
\seclab{sec:avd}

Here, we are given a set $\PntSet$ of $n$ points in $\Re^d$, and our
purpose is to build an \AVD that approximates the $k$-\ANN distance,
while using (roughly) $O(n/k)$ space.

\subsection{Construction}

\subsubsection{Preprocessing}
\begin{figure*}[t]
    \centerline{
       \includegraphics{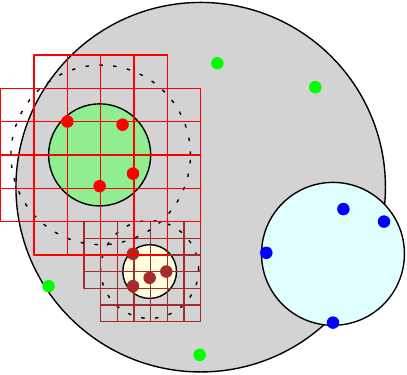}
    }
    \caption{Quorum clustering, immediate environs and grids. }
    \figlab{fig:grid}
\end{figure*}

\SaveIndent
\begin{compactenum}[(A)]
    \RestoreIndent
    \item Compute a quorum clustering for $\PntSet$ using
    \lemref{q:clustering}. Let the list of balls returned be $\ballA_1
    = \pth{\ctrA_1,\radA_1}, \dots, \ballA_{n/k} =
    \pth{\ctrA_{n/k},\radA_{n/k}}$.
    
    \item Compute an exponential grid around each quorum
    cluster. Specifically, let
    \begin{align}
        \ds \CellSetA =\, \bigcup_{i = 1}^{n/k} \;\;\bigcup_{j =0}^{
           \ceiling{\log\pth[]{32/\eps} + 1}}
        \gridSet{\ball{\ctrA_i}{2^j \radA_i} }{ \frac{\eps}{\constA d}
           2^j \radA_i }%
        \eqlab{clusters:around:q}
    \end{align}
    be the set of grid cells covering the quorum clusters and their
    immediate environ, where $\constA$ is a sufficiently large
    constant, see \figref{fig:grid}.
    
    \item Intuitively, $\CellSetA$ takes care of the region of space
    immediately next to a quorum cluster%
    \footnote{That is, intuitively, if the query point falls into one
       of the grid cells of $\CellSetA$, we can answer a query in
       constant time.}.  For the other regions of space, we can apply
    a construction of an approximate Voronoi diagram for
    the centers of the clusters (the details are somewhat more
    involved). To this end, lift the quorum clusters into points in
    $\Re^{d+1}$, as follows
    \begin{align*}
        \mapped{\BallSet} = \brc{\mapped{\ballA_1}, \dots,
           \mapped{\ballA_{n/k}}},
    \end{align*}
    where $\mapped{\ballA_i} = \pth{\ctrA_i,\radA_i} \in \Re^{d+1}$,
    for $i=1,\ldots, n/k$.  Note, that all points in
    $\mapped{\BallSet}$ belong to $\mapped{\Scube} = [0,1]^{d+1}$ by
    \assumpref{points:in:cube}.  Now build a $(1+\eps/8)$-\AVD for
    $\mapped{\BallSet}$ using the algorithm of Arya and Malamatos
    \cite{am-lsavd-02}. The \AVD construction provides a list of
    canonical cubes covering $[0,1]^{d+1}$ such that in the smallest
    cube containing the query point, the associated point of
    $\mapped{\BallSet}$, is a $(1+\eps/8)$-\ANN to the query
    point. (Note, that these cubes are not necessarily disjoint. In
    particular, the smallest cube containing the query point $\query$
    is the one that determines the assigned approximate nearest
    neighbor to $\query$.)

    Clip this collection of cubes to the hyperplane $x_{d+1} = 0$
    (i.e., throw away cubes that do not have a face on this
    hyperplane). For a cube $\cellA$ in this collection, denote by
    $\repX{\cellA}$, the point of $\mapped{\BallSet}$ assigned to it.
    Let $\CellSetC$ be this resulting set of canonical $d$-dimensional
    cubes.

    \item Let $\CellSetAVD$ be the space decomposition resulting from
    overlaying the two collection of cubes, i.e. $\CellSetA$ and
    $\CellSetC$.  Formally, we compute a compressed quadtree $\QTree$
    that has all the canonical cubes of $\CellSetA$ and $\CellSetC$ as
    nodes, and $\CellSetAVD$ is the resulting decomposition of space
    into cells. One can overlay two compressed quadtrees representing
    the two sets in linear time \cite{bhst-sqgqi-10, h-gaa-11}.  Here,
    a cell associated with a leaf is a canonical cube, and a cell
    associated with a compressed node is the set difference of two
    canonical cubes. Each node in this compressed quadtree contains
    two pointers -- to the smallest cube of $\CellSetA$, and 
    to the smallest cube of $\CellSetC$, that contains it. This
    information can be computed by doing a \BFS on the tree.
    
    For each cell $\cellA \in \CellSetAVD$ we store the following.
    \begin{compactenum}[\qquad(I)]
        \item An arbitrary representative point $\repres{\cellA} \in
        \cellA$.
        
        \item The point $\repX{\cellA} \in \mapped{\BallSet}$ that is
        associated with the smallest cell of $\CellSetC$ that contains
        this cell. We also store an arbitrary point, $\pntrepX{\cellA}
        \in \PntSet$, that is one of the $k$ points belonging to the
        cluster specified by $\repX{\cellA}$.
        
        \item A number $\adknn{\repres{\cellA}}$ that satisfies
        $\distPk{\PntSet}{\repres{\cellA}}{k} \leq
        \adknn{\repres{\cellA}} \leq
        (1+\eps/4)\distPk{\PntSet}{\repres{\cellA}}{k}$, and a point
        $\knnrepX{\repres{\cellA}} \in \PntSet$ that realizes this
        distance. In order to compute $\adknn{\repres{\cellA}}$ and
        $\knnrepX{\repres{\cellA}}$ use the data-structure of
        \secref{qtree:algo} (see \thmref{q:tree:main}) or
        the data-structure of Arya \etal \cite{amm-sttas-05}.
    \end{compactenum}
\end{compactenum}

\subsubsection{Answering a query}

Given a query point $\query$, compute the leaf cell (equivalently the
smallest cell) in $\CellSetAVD$ that contains $\query$ by performing a
point-location query in the compressed quadtree $\QTree$.  Let
$\cellA$ be this cell. Return
\begin{align}
    \min\pth{ \MakeBig%
       \distP{\mapped{\query}}{\repX{\cellA}},%
       \,%
       \adknn{\repres{\cellA}} + \dist{\query}{\repres{\cellA}}},
    \eqlab{in:cell}
\end{align}
as the approximate value to $\distPk{\PntSet}{\query}{k}$. 
Return either
$\pntrepX{\cellA}$ or $\knnrepX{\repres{\cellA}}$ depending on which
of the two distances $\distP{\mapped{\query}}{\repX{\cellA}}$ or
$\adknn{\repres{\cellA}} + \dist{\query}{\repres{\cellA}}$ is smaller
(this is the returned approximate value of
$\distPk{\PntSet}{\query}{k}$), as the approximate $k$\th-nearest neighbor.

\subsection{Correctness}

\begin{lemma} %
    \lemlab{algub}%
    Let $\cellA \in \CellSetAVD$ and $\query \in \cellA$. Then the
    number computed by the algorithm is an upper bound on
    $\distPk{\PntSet}{\query}{k}$.
\end{lemma}

\begin{proof}
    By \obsref{1:Lipschitz}, $ \distPk{\PntSet}{\query}{k} \leq
    \distPk{\PntSet}{\repres{\cellA}}{k} +
    \dist{\query}{\repres{\cellA}} \leq \adknn{\repres{\cellA}} +
    \dist{\query}{\repres{\cellA}}$.  Now, let $\repX{\cellA} =(\ctrA,
    \radA)$.  We have, by \lemref{5:approx}, that
    $\distPk{\PntSet}{\query}{k} \leq \dist{\query}{\ctrA} + \radA =
    \distP{\mapped{\query}}{\repX{\cellA}}.  $ As the returned value
    is the minimum of these two numbers, the claim holds.
\end{proof}

\begin{lemma}%
    \lemlab{a:m:easy}%
    Consider any query point $\query \in [0,1]^d$, and let $\cellA$ be
    the smallest cell of $\CellSetAVD$ that contains the query point.
    Then, $\distS{\mapped{\query}}{\mapped{\BallSet}} \leq
    \dist{\mapped{\query}}{\repX{\cellA}} \leq
    (1+\eps/8)\distS{\mapped{\query}}{\mapped{\BallSet}}$.
\end{lemma}

\begin{proof}
    Observe that the space decomposition generated by $\CellSetAVD$ is a
    refinement of the decomposition generated by the Arya and Malamatos
    \cite{am-lsavd-02} \AVD construction, when applied to
    $\mapped{\BallSet}$, and restricted to the $d$ dimensional subspace
    we are interested in (i.e., $x_{d+1}=0$). As such, $\repX{\cellA}$
    is the point returned by the \AVD for this query point
    before the refinement, thus implying the claim.
\end{proof}

\subsubsection{The query point is close to a quorum cluster of the
   right size}

\begin{lemma}%
    \lemlab{lipschitz}%
    Consider a query point $\query$, and let $\cellA \subseteq \Re^d$
    be any set with $\query \in \cellA$, such that $\diameter{\cellA}
    \leq \eps \distPk{\PntSet}{\query}{k}$. Then, for any $\pntA \in
    \cellA$, we have %
    \begin{align*}
        (1-\eps) \distPk{\PntSet}{\query}{k} \leq
        \distPk{\PntSet}{\pntA}{k} \leq
        (1+\eps)\distPk{\PntSet}{\query}{k}.
    \end{align*}%
\end{lemma}%
\begin{proof}
    By \obsref{1:Lipschitz}, we have
    \begin{align*}
        \distPk{\PntSet}{\query}{k}%
        \leq%
        \distPk{\PntSet}{\pntA}{k} + \dist{\pntA}{\query} %
        \leq%
        \distPk{\PntSet}{\pntA}{k} + \diameter{\cellA}%
        \leq%
        \distPk{\PntSet}{\pntA}{k} + \eps \distPk{\PntSet}{\query}{k}.
    \end{align*}
    The other direction follows by a symmetric argument.
\end{proof}

\begin{lemma}%
    \lemlab{small:cell}%
    If the smallest region $\cellA \in \CellSetAVD$ that contains
    $\query$ has diameter $\diameter{\cellA} \leq \eps
    \distPk{\PntSet}{\query}{k}/4$, then the algorithm returns a
    distance which is between $\distPk{\PntSet}{\query}{k}$ and
    $(1+\eps)\distPk{\PntSet}{\query}{k}$.
\end{lemma}
\begin{proof}
    Let $\repres{\cellA}$ be the representative stored with the cell.
    Let $\num$ be the number returned by the algorithm.  By
    \lemref{algub} we have that $\distPk{\PntSet}{\query}{k} \leq
    \num$.  Since the algorithm returns the minimum of two numbers, one
    of which is $\adknn{\repres{\cellA}} +
    \dist{\query}{\repres{\cellA}}$, we have by \lemref{lipschitz},
    \begin{align*}
        \num \;&\leq \adknn{\repres{\cellA}} +
        \dist{\query}{\repres{\cellA}} \leq (1+\eps/4)
        \distPk{\PntSet}{\repres{\cellA}}{k} +
        \dist{\query}{\repres{\cellA}}%
        \\%
        &\leq%
        (1+\eps/4)\pth{\MakeBig \distPk{\PntSet}{\query}{k} + \diameter{\cellA}} +
        \diameter{\cellA}%
        \\%
        &\leq%
        (1+\eps/4)(\distPk{\PntSet}{\query}{k} +
        \eps\distPk{\PntSet}{\query}{k}/4)
        + \eps \distPk{\PntSet}{\query}{k}/4 \\
        &=%
        (1+\eps/4)^2 \distPk{\PntSet}{\query}{k}
        +\eps\distPk{\PntSet}{\query}{k}/4 \leq
        (1+\eps)\distPk{\PntSet}{\query}{k},
    \end{align*}
    establishing the claim.
\end{proof}

\begin{defn}
    Consider a query point $\query \in \Re^d$.  The first quorum
    cluster $\ballA_i = \ball{\ctrA_i}{\radA_i}$ that intersects
    $\ball{\query}{\distPk{\PntSet}{\query}{k}}$ is the \emph{anchor
       cluster} of $\query$. The corresponding \emphi{anchor point} is
    $\pth[]{\ctrA_i,\radA_i} \in \Re^{d+1}$.
    
    \deflab{anchor:p:t}
\end{defn}

\begin{lemma}%
    \lemlab{anchor}%
    For any query point $\query$, we have that
    \begin{compactenum}[(i)]
        \item the anchor point $\pth{\ctrA,\radA}$ is well defined,
        \item $\radA \leq 2\distPk{\PntSet}{\query}{k}$,
        \item for $\ballA = \ball{\ctrA}{\radA}$ we have $\ballA \cap
        \ball{\query}{\distPk{\PntSet}{\query}{k}} \neq \emptyset$,
        and
        \item $\dist{\query}{\ctrA} \leq 3
        \distPk{\PntSet}{\query}{k}$.
    \end{compactenum}
\end{lemma}
\begin{proof}
    Consider the $k$ closest points to $\query$ in $\PntSet$. As
    $\PntSet \subseteq \ballA_1 \cup \dots \cup \ballA_{n/k}$ it must
    be that $\ball{\query}{\distPk{\PntSet}{\query}{k}}$ intersects
    some $\ballA_i$.  Consider the first cluster $\ball{\ctrA}{\radA}$
    in the quorum clustering that intersects
    $\ball{\query}{\distPk{\PntSet}{\query}{k}}$.  Then
    $\pth{\ctrA,\radA}$ is by definition the anchor point and we
    immediately have $\ball{\ctrA}{\radA} \cap
    \ball{\query}{\distPk{\PntSet}{\query}{k}} \neq \emptyset$.  Claim
    (ii) is implied by the proof of \lemref{5:approx}.  Finally, as
    for (iv), we have $\radA \leq 2\distPk{\PntSet}{\query}{k}$ and
    the ball around $\query$ of radius $\distPk{\PntSet}{\query}{k}$
    intersects $\ball{\ctrA}{\radA}$, thus implying that
    $\dist{\query}{\ctrA} \leq \distPk{\PntSet}{\query}{k} + \radA
    \leq 3\distPk{\PntSet}{\query}{k}$.
\end{proof}

\begin{lemma}%
    \lemlab{large:anchor}%
    Consider a query point $\query$. If there is a cluster
    $\ball{\ctrA}{\radA}$ in the quorum clustering computed, such that
    $\dist{\query}{\ctrA} \leq 6 \distPk{\PntSet}{\query}{k}$ and
    $\eps \distPk{\PntSet}{\query}{k}/4 \leq \radA \leq 6
    \distPk{\PntSet}{\query}{k}$, then the output of the algorithm is
    correct.
\end{lemma}

\begin{proof}
    We have
    \begin{align*}
        \frac{32 \radA}{\eps}%
        \geq%
        \frac{32(\eps \distPk{\PntSet}{\query}{k}/4)}{\eps}%
        =%
        8 \distPk{\PntSet}{\query}{k} \geq \dist{\query}{\ctrA}.
    \end{align*}
    Thus, by construction, the expanded environ of the quorum cluster
    $\ball{\ctrA}{\radA}$ contains the query point, see
    \Eqrefpage{clusters:around:q}. Let $j$ be the smallest integer
    such that $2^j \radA \geq \dist{\query}{\ctrA}$. We have that,
    $2^j \radA \leq \max (\radA, 2 \dist{\query}{\ctrA})$.  As such,
    if $\cellA$ is the smallest cell in $\CellSetAVD$ containing the
    query point $\query$, then
    \begin{align*}
        \diameter{\cellA}%
        &\leq %
        \frac{\eps}{\constA d} 2^j \radA%
        \leq %
        \frac{\eps}{\constA d} \cdot \max \pth{ \radA,2
           \dist{\query}{\ctrA} }%
        \leq %
        \frac{\eps}{\constA d} \cdot \max \pth{
           6\distPk{\PntSet}{\query}{k}, 12\distPk{\PntSet}{\query}{k}
           \MakeBig } %
        \\
        &\leq%
        \frac{\eps}{4 d} \distPk{\PntSet}{\query}{k}, %
    \end{align*}
    by \Eqrefpage{clusters:around:q} and if $\constA \geq 48$.  As
    such, $\diameter{\cellA} \leq \eps \distPk{\PntSet}{\query}{k} /
    4$, and the claim follows by \lemref{lipschitz}.
\end{proof}

\subsubsection{The general case}

\begin{lemma}%
    \lemlab{correct}%
    The data-structure constructed above returns
    $(1+\eps)$-approximation to $\distPk{\PntSet}{\query}{k}$, for any
    query point $\query$.
\end{lemma}

\begin{proof}
    Consider the query point $\query$ and its anchor point $(\ctrA,
    \radA)$. By \lemref{anchor}, we have $\radA \leq 2
    \distPk{\PntSet}{\query}{k}$ and $\dist{\query}{\ctrA} \leq 3
    \distPk{\PntSet}{\query}{k}$. This implies that
    \begin{equation}
        \distS{\mapped{\query}}{\mapped{\BallSet}} %
        \leq%
        \dist{\mapped{\query}}{(\ctrA,\radA)}%
        \leq%
        \dist{\query}{\ctrA} + \radA%
        \leq%
        5 \distPk{\PntSet}{\query}{k}.
        \eqlab{no:clue}
    \end{equation}
    Let the returned point, which is a $(1+\eps/8)$-\ANN for
    $\mapped{\query}$ in $\mapped{\BallSet}$, be $(\ctrA_\query,
    \radA_\query) =\repX{\cellA}$, where $\mapped{\query} =
    \pth[]{\query,0}$. We have that
    $\dist{\mapped{\query}}{(\ctrA_\query, \radA_\query)} \leq
    (1+\eps/8)\distS{\mapped{\query}}{\mapped{\BallSet}}\leq 6
    \distPk{\PntSet}{\query}{k}$. In particular,
    $\dist{\query}{\ctrA_\query} \leq 6 \distPk{\PntSet}{\query}{k}$
    and $\radA_\query \leq 6 \distPk{\PntSet}{\query}{k}$.

    Thus, if $\radA_\query \geq \eps \distPk{\PntSet}{\query}{k}/4$ or
    $\radA \geq \eps \distPk{\PntSet}{\query}{k}/4$ we are done,
    by \lemref{large:anchor}.  Otherwise, we have
    \begin{align*}
        \dist{\mapped{\query}}{\pth{\ctrA_\query,\radA_\query}} \leq
        (1 + \eps/8) \dist{\mapped{\query}}{\pth{\ctrA,\radA}},
    \end{align*}
    as $\pth{\ctrA_\query,\radA_\query}$ is a $(1+\eps/8)$
    approximation to $\distS{\mapped{\query}}{\mapped{\BallSet}}$. As
    such,
    \begin{equation}
        \frac{\dist{\mapped{\query}}{\pth{\ctrA_\query,\radA_\query}}}{1 + \eps/8}
        \leq%
        \dist{\mapped{\query}}{\pth{\ctrA,\radA}} 
        \leq%
        \dist{\query}{\ctrA} + \radA.
        \eqlab{rel:to:anchor}
    \end{equation}
    As $\ball{\ctrA}{\radA} \cap
    \ball{\query}{\distPk{\PntSet}{\query}{k}} \neq \emptyset$ we
    have, by the triangle inequality, that
    \begin{equation}
        \dist{\query}{\ctrA} - \radA \leq \distPk{\PntSet}{\query}{k}.
        \eqlab{int:anchor}
    \end{equation}
    
    By \Eqref{rel:to:anchor} and \Eqref{int:anchor} we have
    \begin{align*}
        \frac{\dist{\mapped{\query}}{\pth{\ctrA_\query,\radA_\query}}}{1+\eps/8}
        - 2 \radA%
        \leq%
        \dist{\query}{\ctrA} - \radA%
        \leq%
        \distPk{\PntSet}{\query}{k}.
    \end{align*}
    By the above and as $\max\pth{\radA, \radA_\query} < \eps
    \distPk{\PntSet}{\query}{k}/4$, we have
    \begin{align*}
        \dist{\query}{\ctrA_\query} + \radA_\query &\leq
        \dist{\mapped{\query}}{\pth{\ctrA_\query,\radA_\query}} +
        \radA_\query \leq (1+\eps/8)\pth{\distPk{\PntSet}{\query}{k} +
           2\radA}
        + \radA_\query \\
        &\leq (1+\eps/8)\pth{\distPk{\PntSet}{\query}{k} + \eps
           \distPk{\PntSet}{\query}{k}/2} + \eps
        \distPk{\PntSet}{\query}{k}/4 \leq
        (1+\eps)\distPk{\PntSet}{\query}{k}.
    \end{align*}
    Since the algorithm returns for $\query$ a value that is at most
    $\dist{\query}{\ctrA_\query} + \radA_\query$, the result is
    correct.
\end{proof}

\subsection{The result}

\begin{theorem}%
    \thmlab{ann:main}%
    Given a set $\PntSet$ of $n$ points in $\Re^d$, a number $k \geq
    1$ such that $k \divides n$, and $0 < \eps$ sufficiently small,
    one can preprocess $\PntSet$, in $\ds O \pth{ n \log n +
       \frac{n}{k } C_\eps \log n + \frac{n}{k } C_\eps'}$ time, where
    $C_\eps = O\pth{ \eps^{-d}\log {\eps}^{-1} }$ and $C_\eps' =
    O\pth{ \eps^{-2d+1}\log {\eps}^{-1} }$. The space used by the
    data-structure is $O( C_\eps n/k)$.  This data structure answers
    a $(1+\eps,k)$-\ANN query in $\ds
    \order{\log \frac{n}{k \eps }}$ time.  The data-structure also
    returns a point of $\PntSet$ that is approximately the desired
    $k$-nearest neighbor.
\end{theorem}

\begin{proof}
    Computing the quorum clustering takes time $\order{n \log n}$ by
    \lemref{q:clustering}.  Observe that $\cardin{\CellSetA}
    = \order{\frac{n}{k\eps^d}\log\frac{1}{\eps}}$. From the
    construction of Arya and Malamatos \cite{am-lsavd-02}, 
    we have $\cardin{\CellSetC} =
    \order{\frac{n}{k \eps^{d}}\log \frac{1}{\eps}}$ (note, that since
    we clip the construction to a hyperplane, we get $1/\eps^d$ in the
    bound and not $1/\eps^{d+1}$). A careful implementation of this
    stage takes time $O\pth{ n \log n + \cardin{\CellSetAVD}\pth{\log
          n + \frac{1}{\eps^{d-1}}}}$. Overlaying the two compressed
    quadtrees representing them takes linear time in their size, that
    is $O\pth{ \cardin{\CellSetA} + \cardin{\CellSetC} }$.
    
    The most expensive step is to perform the $(1+\eps/4,k)$-\ANN
    query for each cell in the resulting
    decomposition of $\CellSetAVD$, see \Eqrefpage{in:cell} (i.e.,
    computing $\adknn{\repres{\cellA}}$ for each cell $\cellA \in
    \CellSetAVD$). Using the data-structure of \secref{qtree:algo}
    (see \thmref{q:tree:main}) each query takes $O\pth{ \log n +
       1/\eps^{d-1}}$ time (alternatively, we could use the
    data-structure of Arya \etal \cite{amm-sttas-05}), As such, this
    takes
    \begin{align*}
        O\pth{ n \log n + \cardin{\CellSetAVD}\pth{\log n +
              \frac{1}{\eps^{d-1}}}}%
        =%
        O \pth{%
           n \log n + %
           \frac{n}{k \eps^{d}}\log \frac{1}{\eps}%
           \log n +%
           \frac{n}{k \eps^{2d-1}}\log \frac{1}{\eps}%
        }
    \end{align*}
    time, and this bounds the overall construction time.
    
    The query algorithm is a point location query followed
    by an $O(1)$ time computation and
    takes time $\order{\log \pth{\frac{n}{k \eps}}}$.
    
    Finally, one needs to argue that the returned point of $\PntSet$
    is indeed the desired approximate $k$-nearest neighbor. This
    follows by arguing in a similar fashion to the correctness proof;
    the distance to the returned point is a $(1 + \eps)$-
    approximation to the $k$\th-nearest neighbor distance.  We omit the
    tedious but straightforward details.
\end{proof}

\subsubsection{Using a single point for each \AVD cell}

The \AVD generated can be viewed as storing two points in each cell
$\cellA$ of the \AVD. These two points are in $\Re^{d+1}$, and for a
cell $\cellA$, they are
\begin{compactenum}[\quad(i)]
    \item the point $\repX{\cellA} \in \mapped{\BallSet}$, and
    \item the point $\pth{\repres{\cellA}, \adknn{\repres{\cellA}}}$.
\end{compactenum}
The algorithm for $\distPk{\PntSet}{\query}{k}$ can be viewed as
computing the nearest neighbor of $\pth{\query,0}$ to one of the above
two points using the $\normP{\cdot}$ norm to define the
distance. 
Using standard \AVD algorithms we can subdivide each such cell 
$\cellA$ into $O\pth{ 1 /  \eps^{d} \log \eps^{-1} }$ cells 
to answer this query approximately. By using this finer subdivision 
we can have a single point inside each cell for which the closest 
distance is the approximation to $\distPk{\PntSet}{\query}{k}$. This 
incurs an increase by a factor of $O\pth{ 1/\eps^{d} \log \eps^{-1} }$ 
in the number of cells.

\subsection{A generalization -- weighted version of %
   $k$ \ANN}
\seclab{main:weighted}

We consider a generalization of the $(1+\eps,k)$-\ANN
problem.  Specifically, we are given a set of points $\PntSet
\subseteq \Re^d$, a weight $\wt{\pnt} \geq 0$ for each $\pnt \in
\PntSet$, and a number $\eps > 0$.  
Given a query $\query$ and weight $\wx \geq 0$, its
\emphi{$\wx$-\NNTerm} distance to $\PntSet$, is the 
minimum $r$ such that the
closed ball $\ball{\query}{r}$ contains points of $\PntSet$ of total
weight at least $\wx$. Formally, the $\wx$-\NNTerm distance for $\query$ is
\begin{align*}
    \distPk{\PntSet}{\query}{\wx}%
    =%
    \min \brc{r \sep{ \wtX{ \MakeBig \ball{\query}{r} \cap \PntSet}
          \geq \wx }},
\end{align*}
where $\wtX{X} = \sum_{x \in X} \wt{x}$. A
\emphi{$(1+\eps)$-approximate $\wx$-\NNTerm distance} is a distance $\ell$,
such that $(1-\eps)\distPk{\PntSet}{\query}{\wx} \leq \ell \leq
(1+\eps)\distPk{\PntSet}{\query}{\wx}$ and a 
\emphi{$(1+\eps)$-approximate $\wx$-\NNTerm} is a point of $\PntSet$
that realizes such a distance. The 
\emphi{$(1+\eps,\wx)$-\ANN problem} is to preprocess $\PntSet$, 
such that a $(1+\eps)$-approximate $\wx$-\NNTerm
can be computed efficiently for any query point $\query$.

The $(1+\eps,k)$-\ANN problem is the special case $\wt{\pnt} = 1$
for all $\pnt \in \PntSet$ and $\wx = k$. Clearly, the function 
$\distPk{\PntSet}{\cdot}{\wx}$ is also a $1$-Lipschitz function of its
argument. If we are given $\wx$ at the time of preprocessing, 
it can be verified that the $1$-Lipschitz property is enough to
guarantee correctness of
the \AVD construction for the $(1+\eps,k)$-\ANN problem. However, we 
need to compute a $\wx$ quorum clustering, where
now each quorum cluster has weight at least $\wx$. A slight modification
of the algorithm in \lemref{q:clustering} allows this. Moreover, for the
preprocessing step which requires us to solve the $(1+\eps,\wx)$-\ANN
problem for the representative points, one can use the algorithm of 
\secref{qtree:weighted}. We get the following result,
\begin{theorem}%
    \thmlab{q:tree:main:w}%
    Given a set of $n$ weighted points $\PntSet$ in $\Re^d$, a number
    $\wx > 0$ and $0 < \eps$ sufficiently small, one can preprocess
    $\PntSet$ in $\ds O \pth{ n \log n + \frac{\wtX{\PntSet}}{\wx }
       C_\eps \log n + \frac{\wtX{\PntSet}}{\wx } C_\eps'}$ time,
    where $C_\eps = O\pth{ \eps^{-d}\log {\eps}^{-1} }$ and $C_\eps' =
    O\pth{ \eps^{-2d+1}\log {\eps}^{-1} }$ and 
    $\wtX{\PntSet} = \sum_{\pnt \in \PntSet} \wt(\pnt)$. 
    The space used by the
    data-structure is $O( C_\eps \wtX{\PntSet}/\wx)$.  This data
    structure answers a $(1+\eps,\wx)$-\ANN query in $\ds \order{\log
       \frac{\wtX{\PntSet}}{\wx \eps }}$ time.  The data-structure
    also returns a point of $\PntSet$ that is a $(1+\eps)$-approximation
    to the $\wx$-nearest neighbor of the query point.
\end{theorem}

\section{Density estimation}
\seclab{applications}

Given a point set $\PntSet \subseteq \Re^d$, and a query point $\query
\in \Re^d$, consider the point $\pntB(\query) =
(\distPk{\PntSet}{\query}{1}, \dots,$ $\distPk{\PntSet}{\query}{n})$.
This is a point in $\Re^n$, and several problems in Computational
Geometry can be viewed as computing some interesting function of
$\pntB(\query)$. For example, one could view the nearest neighbor
distance as the function that returns the first coordinate of
$\pntB(\query)$. Another motivating example is a geometric version of
discrete density measures from Guibas \etal \cite{gmm-wkd-11}. In
their problem one is interested in computing $g_k(\query) =
\sum_{i=1}^k \distPk{\PntSet}{\query}{i}$. In this section, we show
that a broad class of functions (that include $g_k$), can be
approximated to within $(1 \pm \eps)$, by a data structure requiring
space $\Otilde(n/k)$.

\subsection{Performing point-location in several quadtrees %
   simultaneously}

\begin{lemma}%
    \lemlab{lowest:color}%
    Consider a rooted tree $T$ with $m$ nodes, where the nodes are
    colored by $I$ colors (a node might have several
    colors). Assume that there are $O(m)$ pairs of such
    $(\text{node},\text{color})$ associations. One can preprocess the
    tree in $O(m)$ time and space, such that given a query leaf $v$ of
    $T$, one can report the nodes $v_1, \ldots, v_I$ in $O(I)$
    time. Here, $v_i$ is the lowest node in the tree along the path
    from the root to $v$ that is colored with color $i$.
\end{lemma}

\begin{proof}
    We start with the naive solution -- perform a \DFS on $T$, and
    keep an array $\Array$ of $I$ entries storing the latest node of
    each color encountered so far along the path from the root to the
    current node. Storing a snapshot of this array $\Array$ at each
    node would require $O(m I)$ space. But then one can answer a query
    in $O(I)$ time. As such, the challenge is to reduce the required
    space.
    
    To this end, interpret the \DFS to be a Eulerian traversal of the
    tree. The traversal has length $2m - 2$, and every edge traveled
    contains updates to the array $\Array$. Indeed, if the \DFS
    traverses down from a node $u$ to a child node $w$, the updates
    would be updating all the colors that are stored in $w$, to
    indicate that $w$ is the lowest node for these colors. Similarly,
    if the \DFS goes up from $w$ to $u$, we restore all the colors
    stored in $w$ to their value just before the \DFS visited
    $w$. Now, the \DFS traversal of $T$ becomes a list of $O(m)$
    updates. Each update is still an $O(I)$ operation.  This is
    however a technicality, and can be resolved as follows.  For each
    edge traveled we store the updates for all colors separately, each
    update being for a single color.  Also each update entry stores
    the current node, i.e. the destination of the edge traveled.  The
    total length of the update list is still $O(m)$, as follows from a
    simple charging argument, and the assumption about the number of
    $(\text{node}, \text{color})$ pairs.  We simply charge each
    restore to its corresponding ``forward going'' update, and the
    number of forward going updates is exactly equal to the number of
    $(\text{node},\text{color})$ pairs.  For each leaf we store its
    last location in this list of updates.
    
    So, let $L$ be this list of updates. At each $k$\th update, for
    $k=tI$ for some integer $t$, store a snapshot of the array of
    colors as updated if we scan the list from the beginning till this
    point. Along with this we store the node at this point and 
    auxiliary information allowing us to compute the next update i.e.
    if the snapshot stored is between all updates at this node. 
    Clearly, all these snapshots can be computed in $O(m)$
    time, and require $O( (m/I) I ) = O(m)$ space.
    
    Now, given a query leaf $v$, we go to its location in the list
    $L$, and jump back to the last snapshot stored. We copy this
    snapshot, and then scan the list from the snapshot till the location
    for $v$. This would require re-doing at most $O(I)$ updates, 
    and can be done in $O(I)$ time overall.
\end{proof}

\begin{lemma}%
    \lemlab{sim:point:location}%
    Given $I$ compressed quadtrees $\DS_1, \ldots, \DS_I$ of total
    size $m$ in $\Re^d$, one can preprocess them in $O( m \log I )$
    time, using $O(m)$ space, such that given a query point $\query$,
    one can perform point-location queries in all $I$ quadtrees,
    simultaneously for $\query$, in $O( \log m + I)$ time.
\end{lemma}
\begin{proof}
    Overlay all these compressed quadtrees together. Overlaying $I$ quadtrees 
    is equivalent to
    merging $I$ sorted lists \cite{h-gaa-11} and can be done in $O\pth{ m
       \log I}$ time. Let $\DS$ denote the resulting compressed
    quadtree. Note that any node of $\DS_i$, for $i=1,\ldots, I$,
    must be a node in $\DS$.
    
    Given a query point $\query$, we need to extract the $I$ nodes in
    the original quadtrees $\DS_i$, for $i=1,\ldots, I$, that contain
    the query point (these nodes can be compressed nodes).  So, let
    $\cellA$ be the leaf node of $\DS$ containing the query point
    $\query$.  Consider the path $\pi$ from the root to the node 
    $\cellA$. We are interested in the lowest node of $\pi$ that
    belongs to $\DS_i$, for $i=1,\ldots, I$. To this end, color all
    the nodes of $\DS_i$ that appear in $\DS$, by color $i$, for
    $i=1,\ldots, I$. Now, we build the data-structure of
    \lemref{lowest:color} for $\DS$. We can use this data-structure to
    answer the desired query in $O(I)$ time.
\end{proof}
\subsection{Slowly growing functions}

\begin{figure*}[t]
    \centerline{
       \begin{tabular}{|l|l|}
           \hline
           $\FSF\MakeVBig$ & The class of slowly growing functions, 
           see \defref{slow}.%
           \\\hline%
           $f\MakeVBig$ &
           A function in $\FSF$ or a monotonic increasing function from
           $\Re$ to $\Re^+$.\\\hline
           $\DDM{k}{f}(\query)$ & %
           \padbox{$\sum_{i=1}^k f\pth{\MakeSBig 
                 \distPk{\PntSet}{\query}{i}}$}%
           \\\hline
           $\aDDM{k}{f}(\query) \MakeVBig$ 
           &
           \padbox{$\sum_{i=\ceiling{k \eps / 8}}^{k}
              f\pth{\distPk{\PntSet}{\query}{i}}$}
           \\\hline%
           $\IdxSet$ & \padbox{$\IdxSet \subseteq \brc{ \MakeSBig\!
                 \ceil{k\eps/8}, \dots, k}$, is a coreset,}%
           see \lemref{di-coreset}.%
           \\\hline%
           \padbox{$w_i, i \in \IdxSet$} & $w_i \geq 0$ are associated 
           weights for coreset elements.%
           \\\hline%
           $\ADDM{k}{f}(\query)$ &%
           \padbox{$\sum_{i \in \IdxSet} w_i 
              f\pth{\distPk{\PntSet}{\query}{i}}$}
           \\\hline
       \end{tabular}
    }
    \caption{Notations used.}
    \figlab{notn:table}
\end{figure*}

\begin{defn}
    \deflab{slow}%
    A monotonic increasing function $f : \Re^+ \to \Re$ is
    \emphi{slowly growing} if there is a constant $\constB > 0$, such
    that for $\eps$ sufficiently small, we have $(1 - \eps) f(x) \leq
    f((1-\eps/\constB) x) \leq f((1+\eps/\constB)x) \leq
    (1+\eps)f(x)$, for all $x \in \Re^+$. The constant $\constB$ is
    the \emphi{growth constant} of $f$. The family of slowly growing
    functions is denoted by $\FSF$.
\end{defn}

Clearly, $\FSF$ includes polynomial functions,
but it does not include, for example, the function $e^x$.  We assume
that given $x$, one can evaluate the function $f(x)$ in constant time.
In this section, using the \AVD construction 
of \secref{sec:avd}, we show how to approximate any function $\DDM{k}{f}(\cdot)$ 
that can be expressed as
\begin{align*}
    \DDM{k}{f}(\query) = \sum_{i=1}^k f\pth{ \MakeSBig
       \distPk{\PntSet}{\query}{i}},
\end{align*}
where $f \in \FSF$.  See \figref{notn:table} for a summary of the
notations used in this section.
\begin{lemma}%
    \lemlab{easy-obs}%
    Let $f : \Re \to \Re^+$ be a monotonic increasing function. Now,
    let $\displaystyle \aDDM{k}{f}(\query)=\sum_{i=\ceiling{k \eps /
          8}}^{k} f\pth{\distPk{\PntSet}{\query}{i}}$. Then, for any
    query point $\query$, we have that $\ds \aDDM{k}{f}(\query) \leq
    \DDM{k}{f}(\query) \leq (1+\eps/4) \aDDM{k}{f}(\query)$, where
    $\DDM{k}{f}(\query) =  \sum_{i=1}^k f\pth{ \MakeSBig
       \distPk{\PntSet}{\query}{i}}$. 
\end{lemma}

\begin{proof}
    The first inequality is obvious. As for the second inequality,
    observe that $\distPk{\PntSet}{\query}{i}$ is a monotonically
    increasing function of $i$, and so is
    $f\pth{\distPk{\PntSet}{\query}{i}}$. We are dropping the smallest
    $k(\eps/8)$ terms of the summation $\DDM{k}{f}(\query)$ that is
    made out of $k$ terms. As such, the claim follows.
\end{proof}

The next lemma exploits a coreset construction, so that we have to
evaluate only few terms of the summation.

\begin{lemma}%
    \lemlab{di-coreset}%
    Let $f : \Re \to \Re^+$ be a monotonic increasing function.  There
    is a set of indices $\IdxSet \subseteq \brc{ \MakeSBig\!
       \ceil{k\eps/8}, \dots, k}$, and integer weights $w_i \geq 0$,
    for $i \in \IdxSet$, such that:%
    \smallskip%
    \begin{compactenum}[\quad(A)]
        \item $\cardin{\IdxSet} = \order{\frac{\log k}{\eps}}$.
        \item For any query point $\query$, we have that
        $\ADDM{k}{f}\pth{\query} = \sum_{i \in \IdxSet} w_i
        f\pth{\distPk{\PntSet}{\query}{i}}$ is a good estimate for
        $\aDDM{k}{f}\pth{\query}$; that is,
        $(1-\eps/4)\ADDM{k}{f}\pth{\query} \leq
        \aDDM{k}{f}\pth{\query} \leq
        (1+\eps/4)\ADDM{k}{f}\pth{\query}$, where
       $\aDDM{k}{f}\pth{\query} = \sum_{i=\ceiling{k \eps /
          8}}^{k} f\pth{\distPk{\PntSet}{\query}{i}}$.
    \end{compactenum}
    \smallskip%
    Furthermore, the set $\IdxSet$ can be computed in $O\pth{
       \cardin{\IdxSet}}$ time.
\end{lemma}

\begin{proof}
    Given a query point $\query$ consider the function $g_\query :
    \brc{1,2,\dots,n} \to \Re^+$ defined as $g_\query(i) =
    f\pth{\MakeSBig \distPk{\PntSet}{\query}{i}}$. Clearly, since $f
    \in \FSF$, it follows that $g_\query$ is a monotonic increasing
    function.  The existence of $\IdxSet$ follows from Lemma $3.2$ in
    {Har-Peled}'s paper \cite{h-cdic-06}, as applied to $(1\pm
    \eps/4)$-approximating the function $\aDDM{k}{f}(\query) =
    \sum_{i=\ceiling{k \eps / 8}}^{k}
    f\pth{\distPk{\PntSet}{\query}{i}}$; that is,
    $(1-\eps/4)\ADDM{k}{f}(\query) \leq \aDDM{k}{f}\pth{\query} \leq
    (1+\eps/4)\ADDM{k}{f}(\query)$.
\end{proof}

\subsection{The data-structure}
We are given a set of $n$ points $\PntSet \subseteq \Re^d$, a function
$f \in \FSF$, an integer $k$ with $1 \leq k \leq n$, and $\eps>0$ sufficiently
small. We describe how to build a data-structure to approximate
$\DDM{k}{f}(\query) = \sum_{i=1}^k f\pth{ \MakeSBig
       \distPk{\PntSet}{\query}{i}}$.
\subsubsection{Construction}
In the following, let $\alpha = 4 \constB$, 
where $\constB$ is the growth constant of $f$ (see
\defref{slow}).  Consider the coreset $\IdxSet$ from
\lemref{di-coreset}.  For each $i \in \IdxSet$ we compute, using
\thmref{ann:main}, a data-structure (i.e., a compressed quadtree)
$\DS_i$ for answering $(1+\eps/\alpha,i)$-\ANN queries for $\PntSet$. 
We then overlay all these
quadtrees into a single quadtree, using \lemref{sim:point:location}.

\paragraph{Answering a Query.}
Given a query point $\query$, perform a simultaneous point-location
query in $\DS_1, \ldots, \DS_I$, by using $\DS$, as described in
\lemref{sim:point:location}.  This results in a $(1+\eps/\alpha)$
approximation $\numA_i$ to $\distPk{\PntSet}{\query}{i}$, for $i \in
\IdxSet$, and takes $O( \log m + I)$ time, where $m$ is the size of
$\DS$, and $I = \cardin{\IdxSet}$.  We return $\xi = \sum_{i \in
   \IdxSet} w_i f\pth{\numA_i}$, where $w_i$ is the weight associated
with the index $i$ of the coreset of \lemref{di-coreset}.

\paragraph{Bounding the quality of approximation.}
We only prove the upper bound on $\xi$. The proof for the lower
bound is similar.  As the $\numA_i$ are $(1 \pm \eps/\alpha)$
approximations to $\distPk{\PntSet}{\query}{i}$ we have,
$(1-\eps/\alpha)\numA_i \leq \distPk{\PntSet}{\query}{i}$, for $i \in
\IdxSet$, and it follows from definitions that,
\begin{align*}
    (1-\eps/4)w_if(\numA_i)%
    \leq%
    w_if\pth{\MakeBig (1-\eps/\alpha)\numA_i}%
    \leq%
    w_i f\pth{ \distPk{\PntSet}{\query}{i}},
\end{align*}
for $i \in \IdxSet$.  Therefore,
\begin{equation}
    \eqlab{eq:1}%
    (1 - \eps/4) \xi%
    = %
    (1-\eps/4)
    \sum_{i \in \IdxSet} w_i f(\numA_i)%
    \leq%
    \sum_{i \in \IdxSet}
    w_i f\pth{ \distPk{\PntSet}{\query}{i}} = \ADDM{k}{f}(\query).
\end{equation}
Using \Eqref{eq:1} and \lemref{di-coreset} it follows that,
\begin{equation}%
    \eqlab{eq:2}%
    (1-\eps/4)^2 \xi \leq (1-\eps/4)
    \ADDM{k}{f}(\query) \leq \aDDM{k}{f}(\query).
\end{equation}
Finally, by \Eqref{eq:2} and \lemref{easy-obs} we have,
\begin{align*}
    (1-\eps/4)^2 \xi \leq \aDDM{k}{f}(\query) \leq
    \DDM{k}{f}(\query).
\end{align*}
Therefore we have, $(1-\eps) \xi \leq (1-\eps/4)^2 \xi \leq
\DDM{k}{f}(\query)$, as desired.

\paragraph{Preprocessing space and time analysis.}
We have that $I = \cardin{\IdxSet} = O\pth{ \eps^{-1} \log k}$.  Let
$C_{x}= \order{x^{-d}\log x^{-1}}$.  By \thmref{ann:main} the total
size of all the $\DS_i$s (and thus the size of the resulting
data-structure) is
\begin{equation}%
    \eqlab{eq:3}%
    S = \sum_{i \in \IdxSet} \order{C_{\eps/\alpha} \frac{n}{i}} =
    \order{C_{\eps/\alpha} \frac{n \log k}{k \eps^2}}.
\end{equation}
Indeed, the maximum of the terms involving $n/i$ is
$\order{\nfrac{n}{k \eps}}$ and $I = \order{\eps^{-1} \log k}$.  By
\thmref{ann:main} the total time taken to construct all the $\DS_i$ is
\begin{align*}
    \sum_{i \in \IdxSet} \order{n \log n + \frac{n}{i} C_{\eps/\alpha}
       \log n + \frac{n}{i} C'_{\eps/\alpha}}%
    =%
    \order{\frac{n \log n \log k}{\eps} + \frac{n\log n\log k}{k
          \eps^2} C_{\eps/\alpha} + \frac{n \log k}{k \eps^2}
       C'_{\eps/\alpha}},
\end{align*}
where $C'_{x} = \order{x^{-2d+1}\log x^{-1}}$.  The time to construct
the final quadtree is $\order{S \log I}$, but this is subsumed by the
construction time above.

\subsubsection{The result}
Summarizing the above, we get the following result.
\begin{theorem}
    \thmlab{thm:appln:ddm}%
    Let $\PntSet$ be a set of $n$ points in $\Re^d$.  Given any slowly
    growing, monotonic increasing function $f$
    (i.e $f \in \FSF$, see \defref{slow}), an integer $k$ with 
    $1 \leq k \leq n$, and $\eps \in (0,1)$,
    one can build a data-structure to approximate
    $\DDM{k}{f}(\cdot)$. Specifically, we have:
    \begin{compactenum}[\qquad\rm (A)]
        \item The construction time is $\order{ C_1 n \log n \log k}$,
        where $C_1 = \order{\eps^{-2d-1}\log \eps^{-1}}$.
        \item The space used is $\ds \order{C_{2} \frac{n }{k } \log
           k}$, where $C_{2}= \order{\eps^{-d-2}\log \eps^{-1}}$.
        \item For any query point $\query$, the data-structure
        computes a number $\xi$, such that $(1-\eps) \xi \leq
        \DDM{k}{f}(\query) \leq (1+\eps) \xi$, where
        $\DDM{k}{f}(\query) = \sum_{i=1}^k
        f(\distPk{\PntSet}{\query}{i})$.
        \item The query time is $\ds \order{\log n + \frac{\log
              k}{\eps}}$.
    \end{compactenum}
    (The $O$ notation here hides constants that depend on $f$.)
\end{theorem}

\section{\ANN queries where $k$ and $\eps$ %
   are part of the query}
\seclab{qtree:algo}

Given a set $\PntSet$ of $n$ points in $\Re^d$, we present a
data-structure for answering $(1+\eps,k)$-\ANN queries, 
in time $\order{\log n +
   1/\eps^{d-1}}$. Here $k$ and $\eps$ are not known
during the preprocessing stage, but are specified during query
time. In particular, different queries
can use different values of $k$ and $\eps$. Unlike our main result,
this data-structure requires linear space, and the amount of space used 
is independent of $k$
and $\eps$. Previous data-structures required knowing $\eps$ in
advance \cite{amm-sttas-05}.

\subsection{Rough approximation}

Observe that a fast constant approximation to
$\distPk{\PntSet}{\query}{k}$ is implied by \thmref{const:main} if $k$
is known in advance.  We describe a polynomial approximation when $k$
is not available during preprocessing.  We sketch the main ideas; our
argument closely follows the exposition in {Har-Peled}'s book
\cite{h-gaa-11}.

\begin{lemma}%
    \lemlab{polyapprox}%
    Given a set $\PntSet$ of $n$ points in $\Re^d$, one can preprocess
    it, in $\order{n \log n}$ time, such that given any query point
    $\query$ and $k$ with $1 \leq k \leq n$, one can find, in $O(\log n)$ time, 
    a number $\approxkd$
    satisfying $\distPk{\PntSet}{\query}{k} \leq \approxkd \leq n^{c}
    \distPk{\PntSet}{\query}{k}$. The result
    is correct with high probability i.e. at least $1 - 1/n^{c-2}$, 
    where $c$ is an arbitrary constant.
\end{lemma}

\begin{proof}
    By an appropriate scaling and translation ensure that $\PntSet
    \subseteq [1/2,3/4]^d$. Consider a compressed quadtree
    decomposition $\QTree$ of $\rvector + [0,1]^d$ for $\PntSet$,
    whose shift $\rvector$ is a random vector in $[0,1/2]^d$. By a 
    bottom-up traversal, compute, for each
    node $v$ of $\QTree$, the axis parallel bounding box $B_v$ of
    the subset of $\PntSet$ stored in its subtree, and the number
    of those points. 
    
    Given a query point $\query \in [1/2,3/4]^d$, locate the
    lowest node $\qtreenode$ of $\QTree$ whose region contains
    $\query$ (this takes $O( \log n)$ time, see \cite{h-gaa-11}). By
    performing a binary search on the root to $\qtreenode$ path 
    locate the lowest node $\qtreenode_k$ whose subtree contains
    $k$ or more points from $\PntSet$.  The algorithm returns
    $\approxkd$, the distance of the query point to the furthest point
    of $B_{\qtreenode_k}$, as the approximate distance.
    
    To see that the quality of approximation is as claimed, 
    consider the ball
    $\ballA$ centered at $\query$ with radius $r =
    \distPk{\PntSet}{\query}{k}$. Next,
    consider the smallest canonical grid having side length $\alpha
    \geq n^{c-1} r$ (thus, $\alpha \leq 2n^{c-1} r$).  Randomly
    translating this grid, we have with probability $\geq 1 - 2rd/
    \alpha \geq 1 - 1/n^{c-2}$, that the ball $\ballA$ is contained inside a
    canonical cell $\cellA$ of this grid. This implies that
    the diameter of $B_{\qtreenode_k}$ is bounded by $\sqrt{d}\alpha$,
    Indeed, if the cell of ${\qtreenode_k}$ is contained in $\cellA$,
    then this clearly holds. Otherwise, if $\cellA$ is contained in
    the cell ${\qtreenode_k}$, then ${\qtreenode_k}$ must be a
    compressed node, the inner portion of its cell is contained in
    $\cellA$, and the outer portion of the cell can not contain any
    point of $\PntSet$. As such, the claim holds.
    
    Moreover, for the returned distance $\approxkd$, we have that
    \begin{align*}
        r%
        =%
        \distPk{\PntSet}{\query}{k} %
        \leq%
        \approxkd \leq \diameter{B_{\qtreenode_k}} + r%
        \leq%
        \sqrt{d} \alpha + r%
        \leq%
        \sqrt{d} 2n^{c-1} r + r \leq n^{c} r.
    \end{align*}
    \aftermathA
\end{proof}

An alternative to the argument used in \lemref{polyapprox}, is to use
two shifted quadtrees, and return the smaller distance returned by the
two trees. It is not hard to argue that in expectation the returned
distance is an $O(1)$-approximation to the desired distance (which then
implies the desired result via Markov's inequality). One can also
derandomize the shifted quadtrees and use $d+1$ quadtrees
instead \cite{h-gaa-11}.

We next show how to refine this approximation.

\begin{lemma}%
    \lemlab{app:refine}%
    Given a set $\PntSet$ of $n$ points in $\Re^d$, one can preprocess
    it in $\order{n \log n}$ time, so that given a query point
    $\query$, one can output a number $\beta$
    satisfying, $\distPk{\PntSet}{\query}{k} \leq \beta \leq
    (1+\eps)\distPk{\PntSet}{\query}{k}$, in $\order{\log n +
       1/\eps^{d-1}}$ time. Furthermore, one can return a point $\pnt
    \in \PntSet$ such that $(1-\eps)\distPk{\PntSet}{\query}{k} \leq
    \dist{\query}{\pnt} \leq (1+\eps)\distPk{\PntSet}{\query}{k}$.
\end{lemma}

\begin{proof}
    Assume that $\PntSet \cup \brc{\query} \subseteq
    [1/2,1/2+1/n]^d$.  The algorithm of \lemref{polyapprox} returns the
    distance $\approxkd$ between $\query$ and some point of $\PntSet$;
    as such we have,
    $\distPk{\PntSet}{\query}{k} \leq \approxkd \leq n^{O(1)}
    \distPk{\PntSet}{\query}{k} \leq \diameter{\PntSet \cup
       \brc{\query}} \leq d/n$.  We start with a compressed quadtree
    for $\PntSet$ having $\Scube = [0,1]^d$ as the root. We look at
    the set of canonical cells $X_0$ with side length at least
    $\approxkd$, that intersect the ball $\ball{\query}{\approxkd}$.
    Clearly, the $k$\th nearest neighbor of $\query$ lies
    in this set of cubes. The set $X_0$ can be computed in $O\pth{
       \cardin{X_0} \log n}$ time using cell queries \cite{h-gaa-11}.

    For each node $v$ in the compressed quadtree there is a
    \emphi{level} associated with it. This is $\levelX{v} = \log_2
    \mathrm{sidelength}(\cellA_v)$. The root has level $0$ and it
    decreases as we go down the compressed quadtree. Intuitively,
    $-\levelX{v}$ is the depth of the node if it was a node in a regular
    quadtree.
    
    We maintain a queue of such canonical grid cells. Each step in the
    search consists of replacing cells in the current level with their
    children in the quadtree, and deciding if we want to descend
    a level. In the $i$\th iteration, we replace every node of
    $X_{i-1}$ by its children in the next level, and put them into the
    set $X_i$.
    
    We then update our estimate of
    $\distPk{\PntSet}{\query}{k}$. Initially, we set $I_0 =
    [\lo_0,\hi_0] =[0, \approxkd]$. For every node $v\in X_i$, we
    compute the closest and furthest point of its cube (that is the
    cell of this node) from the query point (this can be done in
    $O(1)$ time). This specifies a collection of intervals $I_v$ one
    for each node $v \in X_i$.  Let $n_v$ denote the number of points
    stored in the subtree of $v$. For a real number $x$, let $L(x),
    M(x), R(x)$ denote the total number of points in the intervals,
    that are to the left of $x$, contains $x$, and are to the right of
    $x$, respectively. Using median selection, one can compute in
    linear time (in the number of nodes of $X_i$) the minimum $x$ such
    that $L(x) \geq k$. Let this value be $\hi_i$. Similarly, in
    linear time, compute the minimum $x$ such that $L(x) + M(x) \geq
    k$, and let this value be $\lo_i$. Clearly, the desired distance
    is in the interval $I_i = [\lo_i, \hi_i]$.
    
    The algorithm now iterates over $v \in X_i$. If $I_v$ is strictly
    to the left of $\lo_i$, $v$ is discarded
    (it is too close to the query and can not contain the $k$\th
    nearest neighbor), setting $k \leftarrow k - n_v$. Similarly, if
    $I_v$ is to the right of $\hi_i$ it can be thrown away. The algorithm 
    then moves to the next iteration.
    
    The algorithm stops as soon as the diameter of all the cells of
    $X_i$ is smaller than $(\eps/8)\lo_i$.  A
    representative point is chosen from each node of $X_i$ (each node of the
    quadtree has an arbitrary representative point precomputed for it out
    of the subset of points stored in its subtree), and 
    the furthest such point is returned as the $(1+\eps)$- 
    approximate $k$ nearest
    neighbor. To see that the returned answer is indeed correct,
    observe that $\lo_i \leq \distPk{\PntSet}{\query}{k} \leq \hi_i$
    and $\hi_i - \lo_i \leq (\eps/8)\lo_i$, which implies the
    claim. The distance of the returned point from $\query$ is
    in the interval $[\alpha,\beta]$, where $\alpha = \lo_i -
    (\eps/8)\lo_i$ and $\beta = \hi_i \leq \lo_i + (\eps/8)\lo_i \leq
    (1+\eps/2)(1-\eps/8)\lo_i \leq (1+\eps/2)\alpha$. This interval
    also contains $\distPk{\PntSet}{\query}{k}$. As such, $\beta$ is
    indeed the required approximation.
    
    Since we are working with compressed quadtrees, a child node might
    be many levels below the level of its parent. In particular, if a
    node's level is below the current level, we freeze it and just
    move it on the set of the next level. We replace it by its
    children only when its level has been reached.
    
    The running time is clearly $O\pth{
       \cardin{X_0}\log n + \sum_i \cardin{X_i}}$. Let
    $\Delta_i$ be the diameter of the cells in the level being handled
    in the $i$\th iteration. Clearly, we have that $\hi_i \leq \lo_i +
    \Delta_i$. All the cells of $X_i$ that survive must intersect the
    ring with inner and outer radii $\lo_i$ and $\hi_i$ respectively, 
    around $\query$. By a
    simple packing argument, $\cardin{X_i} \leq n_i = O\pth{
       (\lo_i/\Delta_i + 1)^{d-1} }$. As long as $\Delta_i \geq
    \distPk{\PntSet}{\query}{k}$, we have that $n_i=O(1)$, as $\lo_i
    \leq \distPk{\PntSet}{\query}{k}$. This clearly holds for the
    first $O( \log n)$ iterations. It can be verified that once
    this no longer holds, the algorithm performs at most $\ceil{
       \log_2 (1/\eps)} +O(1)$ additional iterations, as then
    $\Delta_i \leq (\eps/16)\distPk{\PntSet}{\query}{k}$ and the
    algorithm stops. Clearly, the $n_i$s in this
    range can grow exponentially, but the last one is
    $O(1/\eps^{d-1})$. This implies that $\sum_i \cardin{X_i} =
    O\pth{\log n + 1/\eps^{d-1}}$, as desired.
\end{proof}

\subsection{The result}

\begin{theorem}%
    \thmlab{q:tree:main}%
    Given a set $\PntSet$ of $n$ points in $\Re^d$, one can preprocess
    them in $\order{n \log n}$ time, into a data structure of size
    $\order{n}$, such that given a query point $\query$, an integer $k$
    with $1 \leq k \leq n$ and $\eps > 0$ one can compute, in $\order{\log n
       + 1/\eps^{d-1}}$ time, a number $\beta$ such that
    $\distPk{\PntSet}{\query}{k} \leq \beta \leq
    (1+\eps)\distPk{\PntSet}{\query}{k}$. The data-structure also
    returns a point $\pnt \in \PntSet$ such that
    $(1-\eps)\distPk{\PntSet}{\query}{k} \leq \dist{\query}{\pnt} \leq
    (1+\eps)\distPk{\PntSet}{\query}{k}$.
\end{theorem}

\subsection{Weighted version of $(1+\eps,k)$-\ANN}
\seclab{qtree:weighted}

We now consider the weighted version of the $(1+\eps,k)$-\ANN problem
as defined in \secref{main:weighted}. Knowledge of the threshold weight 
$\wx$ is not required at the time of preprocessing.
By a straightforward adaptation of the arguments in this section we get
the following.
\begin{theorem}%
    \thmlab{q:tree:main:w:2}%
    Given a set $\PntSet$ of $n$ weighted points in $\Re^d$ one can
    preprocess them, in $\order{n \log n}$ time, into a data structure
    of size $\order{n}$, such that one can efficiently answer
    $(1+\eps,\wx)$-\ANN queries.
    Here a query  is made out of
    \begin{inparaenum}[(i)]
        \item a query point $\query$, 
        \item a weight $\wx \geq 0$, and
        \item an approximation parameter $\eps > 0$.
    \end{inparaenum}
    Specifically, for such a query, one can compute, in $\order{\log n
       + 1/\eps^{d-1}}$ time, a number $\beta$ such that
    $(1-\eps) \distPk{\PntSet}{\query}{\wx} \leq \beta \leq
    (1+\eps)\distPk{\PntSet}{\query}{\wx}$. The data-structure also
    returns a point $\pnt \in \PntSet$ such that
    $(1-\eps)\distPk{\PntSet}{\query}{\wx} \leq \dist{\query}{\pnt}
    \leq (1+\eps)\distPk{\PntSet}{\query}{\wx}$.
\end{theorem}

\section{Density and distance estimation via sampling}
\seclab{sampling}

In this section, we investigate the ability to approximate density
functions using sampling. Note, that sampling can not handle our basic
proximity result (\thmref{ann:main}), since sampling is indifferent to
geometric error. Nevertheless, one can get meaningful results, that
are complementary to our main result, giving another intuition why it
is possible to have sublinear space when approximating the $k$-\NNTerm
and related density quantities.

\subsection{Answering $\pth{ 1+\eps, (1\pm\eps)k}$-\ANN}

\subsubsection{Relative approximation}

We are given a \emphi{range space} $(\X,\R)$, where $\X$ is a set of
$n$ objects and $\R$ is a collection of subsets of $\X$, called
\emphi{ranges}. In a typical geometric setting, $\X$ is a subset of
some infinite ground set $\GroundSet$ (e.g., $\GroundSet=\Re^d$ and
$\X$ is a finite point set in $\Re^d$), and $\R = \brc{\range \cap \X
   \sep{ \range\in \R_\GroundSet}}$, where $\R_\GroundSet$ is a
collection of subsets (i.e., \emphi{ranges}) of $\GroundSet$ of some
simple shape, such as halfspaces, simplices, balls, etc.

The \emphi{measure} of a range $\range \in \R$, is $ \Measure{\range}
= {\cardin{\range}} /{\cardin{\X}}$, and its \emphi{estimate} by a
subset $Z \subseteq \X$ is $\sMeasureX{\range}{Z} = \cardin{\range
   \cap Z} / \cardin{Z}$.  We are interested in range spaces that have
bounded \VC dimension, see \cite{h-gaa-11}. More specifically, we are
interested in an extension of the classical $\eps$-net and
$\eps$-approximation concepts.

\begin{defn}
    For given parameters $0 < \prob,\eps <1$, a subset $Z\subseteq \X$
    is a \emphi{relative $(\prob, \eps)$-approximation} for $(\X,\R)$
    if, for each $\range \in \R$, we have
    \begin{compactenum}[\quad(i)]
        \item $(1-\eps)\Measure{\range} \leq \sMeasureX{\range}{Z}
        \leq (1+\eps)\Measure{\range}$, if $\Measure{\range} \geq \prob$.
        
        \item $\Measure{\range} - \eps \prob \leq
        \sMeasureX{\range}{Z} \leq \Measure{\range} + \eps \prob$, if
        $\Measure{\range} \leq \prob$.
    \end{compactenum}
\end{defn}

\begin{lemma}[\cite{hs-rag-11, h-gaa-11}]%
    \lemlab{relative}%
    For a range space with \VC dimension $\Dim$, a random sample of
    size $\ds O\pth{ \frac{\Dim}{\eps^2 \prob} \pth{
          \log\frac{1}{\prob} + \log \frac{1}{\BadProb}}}$, is a
    relative $(\prob,\eps)$-approximation with probability $\geq
    1-\BadProb$.
\end{lemma}

\subsubsection{Sampling the $(1\pm \eps)k$-\ANN}

So, let $\PntSet$ be a set of $n$ points in $\Re^d$, $k > 0$ and $\eps
\in (0,1)$, be prespecified parameters. The range space of balls in
$\Re^d$ has \VC dimension $d+1$, as follows by a standard lifting
argument, and Radon's theorem \cite{h-gaa-11}. Set $\prob = k/n$, and
compute a random sample $\RSample$ of size
\begin{align*}
    m%
    =%
    O\pth{ \frac{d+1}{\eps^{2}\prob} \pth{ \log \frac{1}{\prob} + \log
          \frac{1}{\BadProb} } }%
    =%
    O\pth{ \frac{n}{k \eps^2} \log \frac{n}{k \BadProb} }.
\end{align*}
This sample is a relative $(p/2,\eps/2)$-approximation
with probability $\geq 1- \BadProb$, and assume that this indeed
holds.

\paragraph{Answering a $(1\pm \eps)k$-\ANN query.}
Given a query point $\query \in \Re^d$, let $\pntA$ be its
$k'$-\NNTerm in $\RSample$, where $k' = \prob m = (k/n)m = O\pth{
   \eps^{-2} \log \frac{n}{k \BadProb} }$.  Return $\pntA$ as the
desired $(1\pm \eps)k$-\ANN.

\paragraph{Analysis.}
Let $r = \dist{\query}{\pntA}$, and consider the ball $\ballA =
\ball{\query}{r}$.  We have that
\begin{align*}
    \sMeasureX{\ballA}{\RSample}%
    =%
    \frac{\cardin{\ballA \cap \RSample}}{\cardin{\RSample}}%
    =%
    \frac{k'}{m}%
    =%
    \frac{k}{n}.
\end{align*}
If $\Measure{\ballA} = \cardin{\ballA \cap \PntSet} /\cardin{\PntSet}
\leq \prob /2 = (k/n)/2 $, then by the relative approximation
definition, we have that $\Measure{\ballA} - \eps (k/n)/4 \leq k/n
\leq \Measure{\ballA} +\eps (k/n)/4$. But this implies that
$\Measure{\ballA} \geq (3/4) (k/n)$, which is a contradiction.

As such, we have that $\Measure{\ballA} \geq \prob /2$. Again, by the
relative approximation definition, we have that $(1-\eps
/2)\Measure{\range} \leq \sMeasureX{\range}{\RSample} \leq (1+\eps
/2)\Measure{\range}$, and this in turn implies that

\begin{align*}
    (1-\eps) k%
    \leq %
    \frac{n}{1+\eps/2 } \sMeasureX{\range}{\RSample}%
    \leq%
    n \cdot \Measure{\range}%
    =%
    \cardin{\ballA \cap \PntSet}%
    \leq%
    \frac{n}{1-\eps/ 2} \sMeasureX{\range}{\RSample}%
    \leq%
    (1+\eps) k,
\end{align*}
as $\sMeasureX{\ballA}{\RSample} = k/n$.

\paragraph{The result.} 
Of course, there is no reason to compute the exact $k'$-\NNTerm in
$\RSample$. Instead, one can compute the $(1+\eps,k')$-\ANN in
$\RSample$ to the query. In particular, using the data-structure of
\thmref{q:tree:main}, we get the following.

\begin{lemma}%
    \lemlab{random:sample:k:ann}%
    Given a set $\PntSet$ of $n$ points in $\Re^d$, and parameters $k$
    $\eps > 0$, and $\BadProb > 0$. Consider a random sample
    $\RSample$ from $\PntSet$ of size $m= O\pth{ \frac{n}{k \eps^2}
       \log \frac{n}{k \BadProb} }$. One can build a data-structure in
    $O(m \log m)$ time, using $O(m)$ space, such that for any query
    point $\query$, one can compute a $\pth{\MakeBig 1+\eps, (1\pm
       \eps)k}$-\ANN in $\PntSet$, by answering $k'$-\NNTerm or
    $(1+\eps,k')$-\ANN query on $\RSample$, where $k'= O\pth{
       \eps^{-2} \log \frac{n}{k \BadProb} }$.

    Specifically, the query time is $O(\log m + 1/\eps^{d-1})$, and
    the result is correct for all query points with probability $\geq
    1- \BadProb$; that is, for the returned point $\pntA$, we have
    that $(1-\eps)\distPk{\PntSet}{\query}{(1-\eps)k} \leq
    \dist{\query}{\pnt} \leq
    (1+\eps)\distPk{\PntSet}{\query}{(1+\eps)k}$.
\end{lemma}  

\begin{remark}
    \begin{inparaenum}[(A)]
        \item If one plugs the random sample into \thmref{ann:main},
        then one gets a data-structure of size $O\pth{n /\pth{k
              \eps^{O(1)}}}$, that can answer $\pth{ 1+\eps,
           (1\pm\eps)k}$-\ANN in logarithmic time.

        \item Once computed, the data-structure of
        \lemref{random:sample:k:ann} works for approximating any
        $\pth[]{ 1+\eps, (1\pm\eps)t}$-\ANN, for any $t \geq k$, by
        computing the $(1+\eps, t')$-\ANN on $\RSample$, where $t' =
        (t/n) m$.
    \end{inparaenum}
\end{remark}

\subsection{Density estimation via sampling}

\subsubsection{Settings}
\seclab{settings:density}

Let $\PntSet$ be a set of $n$ points in $\Re^\hdim$, and let $k$ be a
parameter.  In the following, for a point $\query$, let $\Pk{\query}$
be the set of $k$ points closest to $\query$ in $\PntSet$.  For such a
query point $\query \in \Re^d$, we are interested in estimating the
quantity
\begin{align}%
    \eqlab{estimate:sq}%
    {F}_1(\query) = \pth[]{\frac{1}{k}\sum_{\pntA \in \Pk{\query}}
       {\distY{\query}{\pntA}}^2 }^{1/2}.
\end{align}
Since we care only about approximation, it is sufficient to
approximate the function without the square root. Formally, a
$\pth[]{1+O(\eps^2)}$-approximation $\alpha$ to
$\pth[]{{F}_1(\pnt)}^2$, yields the approximation $\sqrt{\alpha}$ to
${{F}_1(\pnt)}$, and this is a $(1+\eps)$-approximation to the
original quantity, see \cite[Lemma 4.6]{ahv-aemp-04}. Furthermore, as
in \defrefpage{slow}, we can handle more general functions than
squared distances. However, since we are interested in random
sampling, we have to assume something additional about the
distribution of points.

\begin{defn}%
    \deflab{well:behaved}%
    For a point-set $\PntSet \subseteq \Re^d$, and a parameter $k$,
    the function $f:\Re \rightarrow \Re^+$ is a \emphi{well-behaved}
    distance function, if %
    \smallskip%
    \begin{compactenum}[\quad(i)]
        \item $f$ is monotonically increasing, and
        \item for any point $\query \in \Re^d$, there exists a
        constant $\constC > 0$, such that $f\pth{
           \distPk{\PntSet}{\query}{(3/2)k}} \leq \constC f
        \pth{\distPk{\PntSet}{\query}{k/4}}$.
    \end{compactenum}
    \smallskip%
    A set $\Family$ of functions is \emphi{well-behaved} if the above
    holds for any function in $\Family$ (with the same constant
    $\constA$ for all the functions in $\Family$).
\end{defn}

As such, the target here is to approximate
\begin{align}
    F(\query)%
    =%
    \frac{1}{k}\sum_{\pntA \in \Pk{\query}}
    f\pth{\distY{\query}{\pntA} \MakeBig\! } ,%
    \eqlab{halving}%
\end{align}
where $f(\cdot)$ is a well-behaved distance function. 

\subsubsection{The estimation algorithm}

Let $\RSample$ be a random sample from $\PntSet$ of size $\ds m =
O\pth{ \frac{ d }{\prob \eps^2 } \log \frac{ n}{k \BadProb}}$, where
$\prob = k/n$, and $\BadProb>0$ is a prespecified confidence
parameter. Given a query $\query$, compute the quantity
\begin{align}
    G(\query)%
    =%
    \frac{1}{k'}\sum_{\pntA \in \PkExt{\RSample}{k'}{\query}}
    f\pth{\distY{\query}{\pntA} \MakeBig\! } ,%
    \eqlab{estimate}%
\end{align}
where $k'= \prob m$. Return this as the required estimate to
$F(\query)$, see \Eqref{halving}.

\subsubsection{Analysis}

We claim that this estimate is good, with good probability for all
query points.  Fix a query point $\query \in \Re^d$, and let $\eps >
0$ be the prespecified approximation parameter.  For the sake of
simplicity of exposition, we assume that
$f\pth{\distPk{\PntSet}{\query}{(1+\eps)k}} = k/n$ -- this can be
achieved by dividing $f(\cdot)$ by the right constant, and applying
our analysis to this modified function. In particular,
$f\pth{\distPk{\PntSet}{\query}{i}} \geq k/(\constC n)$, for all $i
\geq k/4$. For any $r \geq 0$, let
\begin{align*}
    h_{\query,r}(\pntA) = \left\{ \begin{array}{ll}
            \frac{n}{k}f\pth{\distY{\query}{\pntA}} &
            \distY{\query}{\pntA} \leq
            r, \\
            0 & \text{otherwise}.
           \end{array}
       \right.
\end{align*}
Consider a value $x\geq 0$. The sublevel set of all points $\pntC$,
such that $h_{\query,r}(\pntC) \leq x$, is the union of
\begin{inparaenum}[(i)]
    \item a ball centered at $\query$, with 
    \item a complement of a ball (also centered at $\query$ of radius
    $r$.
\end{inparaenum}
(i.e., its the complement of a ring.)  This follows as $f$ is a
monotonically increasing function. As such, consider the family of
functions
\begin{align*}
    \Family = \brc{ h_{\query,r}(\cdot) \sep{ \query \in \Re^d, r \geq
          0}}.
\end{align*}
This family has bounded pseudo-dimension (a fancy way to say that the
sublevel sets of the functions in this family have finite \VC
dimension), which is $O(d)$ in this case, as every range is the union
of a ball and a ball complement \cite[Section 5.2.1.1]{h-gaa-11}.
Now, we can rewrite the quantity of interest as
\begin{align}
    F(\query)%
    =%
    \frac{1}{k}\sum_{\pntA \in \Pk{\query}} f\pth{\MakeBig\!
       \distY{\query}{\pntA} }%
    =%
    \frac{1}{n} \sum_{\pntA \in \Pk{\query}} \frac{n}{k}
    f\pth{\MakeBig\!  \distY{\query}{\pntA} }%
    =%
    \frac{1}{n} \sum_{i=1}^n h_{\query,r}\pth{ \pntA_i},
    \eqlab{F:q}
\end{align}
where $r= \distPk{\PntSet}{\query}{k}$ (here $r$ is a function of
$\query$).  Note, that by our normalization of $f$, we have that $
h_{\query,r}\pth{ \pntC } \in [0,1]$, for any $\pntC \in \Re^d$. We
are now ready to deploy a sampling argument. We need a generalization
of $\eps$-approximation due to Li \etal \cite{lls-ibscl-01}, see also
\cite{h-gaa-11}.

\begin{theorem}[\cite{lls-ibscl-01}]
    \thmlab{SAMPLING}%
    Let $\alpha, \nu, \BadProb > 0$ be parameters, let $\RangeSpace =
    \pth{\GroundSet, \Family}$ be a range space, and let $\Family$ be
    a set of functions from $\GroundSet$ to $\pbrcS{0,1}$, such that
    the pseudo-dimension of $\RangeSpace$ is $\Dim$.  For a random
    sample \index{random sample} $\RSample$ (with repetition) from
    $\GroundSet$ of size $\ds O\pth{\frac{1}{\alpha^2 \nu} \pth{ \Dim
          \log \frac{1}{\nu} + \log \frac{1}{\BadProb}}}$, we have
    that
    \[
    \forall g \in \Family \;\;\;\; \DistX{\nu}{\,\Measure{g}, \,
       \sMeasureX{g}{\RSample} \MakeSBig } < \alpha
    \]
    with probability $\geq 1-\BadProb$.
\end{theorem}

Lets try to translate this into human language. In our case,
$\GroundSet = \PntSet$. For the following argument, we fix the query
point $\query$, and the distance $r=\dk{\query}$.  The measure
function is
\begin{align*}
    \Measure{g} = \sum_{\pntA \in \GroundSet} \Prob{\pntA\MakeBig}
    g(\pntA) = F(\query),
\end{align*}
which is the desired quantity if one set $\Prob{\pnt} = 1/n$, and
$g(\pntA) = h_{\query,r}\pth{ \pntA }$ -- see \Eqref{F:q}.  For the
sample $\RSample $, the estimate is
\begin{align*}
    \sMeasureX{g}{\RSample}%
    =%
    \frac{1}{m} \sum_{\pntA \in \RSample} g\pth{\pntA}%
    =%
    \frac{1}{m} \sum_{\pntA \in \RSample} h_{\query,r}\pth{ \pntA },
\end{align*}
where $m=\cardin{\RSample}$.  Now, by the normalization of $f(\cdot)$,
we have that $\Measure{g} = F(\query) \geq k/2\constC n$ and
$\Measure{g} \leq k/n$.  The somewhat mysterious distance function, in
the above theorem, is
\begin{align*}
    \DistX{\nu}{ \rho, \varrho} = \frac{\cardin{\rho-\varrho}}{\rho+
       \varrho+ \nu}.
\end{align*}
Setting 
\begin{align}%
    \eqlab{parameter:setting}%
    \nu = \tfrac{k}{16 \constC n}%
    \qquad \text{ and }\qquad%
    \alpha = \tfrac{\eps}{16},
\end{align}
the condition in the theorem is
\begin{align}%
    \eqlab{amazing}%
    & \forall g \in \Family \;\;\;\; \DistX{\nu}{ \MakeBig
       \,\Measure{g}, \, \sMeasureX{g}{\RSample} \MakeSBig } < \alpha
    \qquad \Longrightarrow \qquad
    \cardin{\MakeBig \Measure{g} - \sMeasureX{g}{\RSample} }%
    <%
    \frac{\eps}{4} \Measure{g},
\end{align}
as an easy but tedious calculation shows.  This is more or less the
desired approximation, except that we do not have $r$ at
hand. Conceptually, the algorithm first estimates $r$, from the
sample, see \Eqref{estimate}, by computing the $k'$\th nearest
neighbor to the query in $\RSample$, and then computes the estimate
using this radius. Formally, let $r' = \distPk{\RSample}{\query}{k'}$,
and observe that as $k' = \prob m = (k/n) m$, we have  
\begin{align*}
    G(\query)%
    &=%
    \frac{1}{k'}\sum_{\pntA \in \PkExt{\RSample}{k'}{\query}}
    f\pth{\distY{\query}{\pntA} \MakeBig\! } %
    =%
    \frac{1}{k'}\cdot \frac{k}{n}\sum_{\pntA \in
       \PkExt{\RSample}{k'}{\query}} \frac{n}{k}
    f\pth{\distY{\query}{\pntA} \MakeBig\! } %
    =%
    \frac{1}{k'}\cdot \frac{k}{n}\sum_{\pntA \in
       \PkExt{\RSample}{k'}{\query}} h_{\query,r'}(\pntA)%
    \\%
    &=%
    \frac{1}{m}\sum_{\pntA \in \RSample} h_{\query,r'}(\pntA).%
\end{align*}
In particular, the error between the algorithm estimate, and theorem
estimate is 
\begin{align*}
    \Error%
    =%
    \cardin{ \MakeBig G(\query) - \sMeasureX{g}{\RSample}}%
    =%
    \cardin{%
       \frac{1}{m}\sum_{\pntA \in \RSample} h_{\query,r'}(\pntA) -
       \frac{1}{m}\sum_{\pntA \in \RSample} h_{\query,r}(\pntA)%
    }.
\end{align*}

Now, by \lemref{relative}, $\RSample$ is a relative
$(\prob/4,\eps/\constD )$-approximation, with probability $\geq
1-\BadProb/10$, where $\constD > 0$ is a sufficiently large constant
(its exact value would follow from our analysis). This implies that
the ball centered at $\query$ of radius $r'$, contains between
$[(1-\eps/\constD)k, (1+\eps/\constD)k]$ points of $\PntSet$. This in
turn implies that number of points of $\RSample$ in the ball of radius
$r'$ centered at $\query$ is in the range
$\pbrc[]{(1-\eps/\constD)^2k', (1+\eps/\constD)^2k'}$. This in turn
implies that the number of ``heavy'' points in the sample $\RSample$
is relatively small. Specifically, the number of points in $\RSample$
that are in the ball of radius $r'$ around $\query$, but not in the
concentric ball of radius $r$ (or vice versa) is
\begin{align*}
    \cardin{ \MakeBig%
       \cardin{\PkExt{\RSample}{r'}{\query}} -
       \cardin{\PkExt{\RSample}{r}{\query}} }%
    \leq (1+\eps/\constD)^2k' - (1-\eps/\constD)^2k'%
    \leq%
    (6
    \eps/\constD)k'.
\end{align*}
By the well-behaveness of $f$, this implies that the contribution of
these points is marginal compared to the ``majority'' of points in
$\RSample$; that is, all the points in $\RSample$ that are the $i$\th
nearest-neighbor to $\query$, for $i=k'/2, \ldots (3/4)k'$, have
weight at least $\alpha/\constC$, where $\alpha$ is the maximum value
of $h_{q,r'}$ on any point of
$\PkExt{\RSample}{(1+\eps)k'}{\query}$. That is, we have
\begin{align*}
    \Delta%
    &=%
    \min \pth{G(\query), \; \sMeasureX{g}{\RSample} \MakeBig }%
    =%
    \frac{1}{m} \min \pth{ \sum_{\pntA \in \RSample}
       h_{\query,r'}(\pntA) , \sum_{\pntA \in \RSample}
       h_{\query,r}(\pntA) }%
    \geq%
    \frac{1}{m} \sum_{\pntA \in \PkExt[]{\RSample}{(3/4)k'}{\query}}
    h_{\query,r'}(\pntA)%
    \\
    &\geq%
    \frac{1}{m} \cdot \frac{k'}{4} \cdot \frac{\alpha}{\constC}%
    =%
    \frac{\alpha k' }{4 m \constC}.
\end{align*}
Similarly,  we have 
\begin{align*}
    \Error &= \frac{1}{m} \cardin{ \sum_{\pntA \in \RSample}
       h_{\query,r'}(\pntA) - \sum_{\pntA \in \RSample}
       h_{\query,r}(\pntA) }%
    \leq%
    \frac{1}{m} \cardin{ \MakeBig%
       \cardin{\PkExt{\RSample}{r'}{\query}} -
       \cardin{\PkExt{\RSample}{r}{\query}} }%
    \cdot \alpha%
    \leq%
    \frac{1}{m} \cdot %
    \frac{6 \eps}{\constD}k'%
    \cdot \alpha%
    \\
    &= %
    \frac{6 \eps \alpha k'}{16m \constD}%
    =%
    \eps \cdot \frac{6 \constC }{4\constD} \cdot \frac{\alpha k'}{4m
       \constC}%
    \leq %
    \frac{\eps}{4} \Delta%
    \leq%
    \frac{\eps}{4} \sMeasureX{g}{\RSample} \,,
\end{align*}
if $\constD \geq 6\constC$.  We thus have that%
\vspace{-0.5cm}%
\begin{align*}
    \cardin{G(\query) - F(\query)}%
    &= %
    \cardin{G(\query) - \Measure{g}}%
    \leq%
    \overbrace{ \cardin{G(\query) - \sMeasureX{g}{\RSample} \MakeBig
       }}^{=\Error} + \cardin{\sMeasureX{g}{\RSample} - \Measure{g}}%
    \\
    &\leq%
    \frac{\eps}{4} \sMeasureX{g}{\RSample} + \frac{\eps}{4}
    \Measure{g}%
    \leq%
    \frac{\eps}{4} \pth{1+\frac{\eps}{4}} \Measure{g} +\frac{\eps}{4}
    \Measure{g}%
    \leq%
    \eps \Measure{g}%
    =%
    \eps F(\query),
\end{align*}
by \Eqref{amazing}.  That is, the returned approximation has small
error.

The above analysis assumed both that the sample $\RSample$ is a
relative $(\prob/4,\eps/\constD)$-\si{approxima}\-\si{tion} (for
balls), and also complies with \thmref{SAMPLING}, for the range space,
where the ranges are a complement of a single ring, for the parameters
set in \Eqref{parameter:setting}. Clearly, both things hold with
probability $\geq 1-\BadProb$, for the size of the sample taken by the
algorithm. Significantly, this holds for all query points.

\subsubsection{The result}

\begin{theorem}
    Let $\PntSet$ be a set of $n$ points in $\Re^\hdim$, $k$, $\eps >
    0$ and $\BadProb > 0$ be parameters. Furthermore, assume that we
    are given a well-behaved function $f(\cdot)$ (see
    \defrefpage{well:behaved}). Let $\RSample$ be a random sample of
    $\PntSet$ of size $\ds m = O\pth{ \frac{ d n }{k \eps^2 } \log
       \frac{ n}{k \BadProb}}$. Then, with probability $\geq
    1-\BadProb$, for all query points $\query \in \Re^\hdim$, we have
    that for the quantity
    \begin{align*}
        F(\query)%
        =%
        \frac{1}{k}\sum_{\pntA \in \Pk{\pnt}} f\pth{\MakeBig\!
           \distY{\query}{\pntA} }%
        \quad \text{ and its estimate }\quad %
        G(\query)%
        =%
        \frac{1}{k'}\sum_{\pntA \in \PkExt{\RSample}{k'}{\query}}
        f\pth{\distY{\query}{\pntA} \MakeBig\! } ,%
    \end{align*}
    we have that $\cardin{F(\query) - G(\query) }\leq \eps F(\query)$,
    where $k'= (k/n) m$.  Here, $ \PkExt{\RSample}{k'}{\query}$
    denotes the set of $k'$ nearest-neighbor to $\query$ in
    $\RSample$.
\end{theorem}

The above theorem implies that one can get a $(1\pm\eps)$
multiplicative approximation to the function $F(\query)$, for all
possible query points, using $O(m)$ space. Furthermore, the above
theorem implies that any reasonable density estimation for a point-set
that has no big gaps, can be done using a sublinear sample size; that
is, a sample of size roughly $O( d n/k)$, which is (surprisingly)
polynomial in the dimension. This result is weaker than
\thmref{thm:appln:ddm}, as far as the family of functions it handle,
but it has the advantage of being of linear size (!) in the dimension.
This compares favorably with the recent result of \Merigot{}
\cite{m-lbfkd-13}, that shows an exponential lower bound
$\Omega\pth{1/\eps^{\hdim}}$ on the complexity of such an
approximation for a specific such distance function, when the
representation used is (essentially) additive weighted Voronoi diagram
(for $k=n/2$). More precisely, the function \Merigot{} studies has the
form of \Eqrefpage{estimate:sq}. However, as pointed out in
\secref{settings:density}, up to squaring the sample size, our result
holds also in this case.

\section{Conclusions}
\seclab{conclusions}

In this paper, we presented a data-structure for answering
$(1+\eps,k)$-\ANN queries in $\Re^d$ where $d$ is a constant. Our
data-structure has the surprising property that the space required is
$\Otilde(n/k)$. One can verify that up to noise this is the best one
can do for this problem. This data-structure also suggests a natural
way of compressing geometric data, such that the resulting sketch can
be used to answer meaningful proximity queries on the original
data. We then used this data-structure to answer various proximity
queries using roughly the same space and query time.  We also
presented a data-structure for answering $(1+\eps,k)$-\ANN queries
where both $k$ and $\eps$ are specified during query time. This
data-structure is simple and practical. Finally, we investigated what
type of density functions can be estimated reliably using random
sampling.

There are many interesting questions for further research.
\begin{compactenum}[\qquad (A)]
    \item In the vein of the authors recent work \cite{hk-annsl-11},
    one can verify that our results extends in a natural way
    to metrics of low doubling dimensions (\cite{hk-annsl-11}
    describes what an approximate Voronoi diagram is for doubling
    metrics). It also seems believable that the result would extend to
    the problem where the data is high dimensional but the queries
    arrive from a low dimensional manifold.
    
    \item It is natural to ask what one can do for this problem in
    high dimensional Euclidean space. In particular, can one get query
    time close to the one required for approximate nearest neighbor
    \cite{im-anntr-98, him-anntr-12}. Of particular interest is getting a query time
    that is sublinear in $k$ and $n$ while having subquadratic space
    and preprocessing time.
    
    \item The dependency on $\eps$ in our data-structures may
    not be optimal. One can probably get space/time tradeoffs, as done by
    Arya \etal \cite{amm-sttan-09}.
\end{compactenum}

\subsection*{Acknowledgments.}

The authors thank Pankaj Agarwal and Kasturi Varadarajan for useful
discussions on the problems studied in this paper.

\bibliographystyle{alpha}%
\bibliography{shortcuts,geometry}%

\newcommand{\etalchar}[1]{$^{#1}$}
 \providecommand{\Merigot}{M{\'{}e}rigot}
  \providecommand{\Matousek}{Matou{\v s}ek} \providecommand{\Erdos}{Erd{\H o}s}
  \providecommand{\Barany}{B{\'a}r{\'a}ny}
  \providecommand{\Bronimman}{Br{\"o}nnimann}
  \providecommand{\Gartner}{G{\"a}rtner} \providecommand{\Badoiu}{B\u{a}doiu}
  \providecommand{\tildegen}{{\protect\raisebox{-0.1cm}
  {\symbol{'176}\hspace{-0.03cm}}}} \providecommand{\SarielWWWPapersAddr}
  {http://www.uiuc.edu/~sariel/papers} \providecommand{\SarielWWWPapers}
  {http://www.uiuc.edu/\tildegen{}sariel/\hspace{0pt}papers}
  \providecommand{\urlSarielPaper}[1]{ \href{\SarielWWWPapersAddr/#1}
  {\SarielWWWPapers{}/#1}}
\begin{thebibliography}{dBHTT10}

\bibitem[AE98]{ae-rsir-98}
P.~K. Agarwal and J.~Erickson.
\newblock Geometric range searching and its relatives.
\newblock In B.~Chazelle, J.~E. Goodman, and R.~Pollack, editors, {\em Advances
  in Discrete and Computational Geometry}. Amer. Math. Soc., 1998.

\bibitem[AHV04]{ahv-aemp-04}
P.~K. Agarwal, S.~{Har-Peled}, and K.~R. Varadarajan.
\newblock Approximating extent measures of points.
\newblock {\em J. Assoc. Comput. Mach.}, 51(4):606--635, 2004.

\bibitem[AI08]{ai-nohaa-08}
A.~Andoni and P.~Indyk.
\newblock Near-optimal hashing algorithms for approximate nearest neighbor in
  high dimensions.
\newblock {\em Commun. ACM}, 51(1):117--122, 2008.

\bibitem[AM93]{am-rsps-93}
P.~K. Agarwal and J.~Matou\v{s}ek.
\newblock Ray shooting and parametric search.
\newblock {\em SIAM J. Comput.}, 22:540--570, 1993.

\bibitem[AM02]{am-lsavd-02}
S.~Arya and T.~Malamatos.
\newblock Linear-size approximate {Voronoi} diagrams.
\newblock In {\em Proc. 13th ACM-SIAM Sympos. Discrete Algs.}, pages 147--155,
  2002.

\bibitem[AMM05]{amm-sttas-05}
S.~Arya, T.~Malamatos, and D.~M. Mount.
\newblock Space-time tradeoffs for approximate spherical range counting.
\newblock In {\em Proc. 16th ACM-SIAM Sympos. Discrete Algs.}, pages 535--544,
  2005.

\bibitem[AMM09]{amm-sttan-09}
S.~Arya, T.~Malamatos, and D.~M. Mount.
\newblock Space-time tradeoffs for approximate nearest neighbor searching.
\newblock {\em J. Assoc. Comput. Mach.}, 57(1):1--54, 2009.

\bibitem[AMN{\etalchar{+}}98]{amnsw-oaann-98}
S.~Arya, D.~M. Mount, N.~S. Netanyahu, R.~Silverman, and A.~Y. Wu.
\newblock An optimal algorithm for approximate nearest neighbor searching in
  fixed dimensions.
\newblock {\em J. Assoc. Comput. Mach.}, 45(6):891--923, 1998.

\bibitem[Aur91]{a-vdsfg-91}
F.~Aurenhammer.
\newblock {Voronoi} diagrams: A survey of a fundamental geometric data
  structure.
\newblock {\em ACM Comput. Surv.}, 23:345--405, 1991.

\bibitem[BSLT00]{bslt-gagem-00}
M.~Bernstein, V.~De Silva, J.~C. Langford, and J.~B. Tenenbaum.
\newblock Graph approximations to geodesics on embedded manifolds, 2000.

\bibitem[CDH{\etalchar{+}}05]{cdhks-gqsa-05}
P.~Carmi, S.~Dolev, S.~{Har-Peled}, M.~J. Katz, and M.~Segal.
\newblock Geographic quorum systems approximations.
\newblock {\em Algorithmica}, 41(4):233--244, 2005.

\bibitem[CH67]{ch-nnpc-67}
T.M. Cover and P.E. Hart.
\newblock Nearest neighbor pattern classification.
\newblock {\em IEEE Transactions on Information Theory}, 13:21--27, 1967.

\bibitem[Cha08]{c-tpfgn-08}
B.~Chazelle.
\newblock Technical perspective: finding a good neighbor, near and fast.
\newblock {\em Commun. ACM}, 51(1):115, 2008.

\bibitem[Cha10]{c-opt-10}
T.~M. Chan.
\newblock Optimal partition trees.
\newblock In {\em Proc. 26th Annu. ACM Sympos. Comput. Geom.}, pages 1--10,
  2010.

\bibitem[CK95]{ck-dmpsa-95}
P.~B. Callahan and S.~R. Kosaraju.
\newblock A decomposition of multidimensional point sets with applications to
  $k$-nearest-neighbors and $n$-body potential fields.
\newblock {\em J. Assoc. Comput. Mach.}, 42:67--90, 1995.

\bibitem[Cla88]{c-racpq-88}
K.~L. Clarkson.
\newblock A randomized algorithm for closest-point queries.
\newblock {\em SIAM J. Comput.}, 17:830--847, 1988.

\bibitem[Cla06]{c-nnsms-06}
K.~L. Clarkson.
\newblock Nearest-neighbor searching and metric space dimensions.
\newblock In G.~Shakhnarovich, T.~Darrell, and P.~Indyk, editors, {\em
  Nearest-Neighbor Methods for Learning and Vision: Theory and Practice}, pages
  15--59. MIT Press, 2006.

\bibitem[CS89]{cs-arscg-89}
K.~L. Clarkson and P.~W. Shor.
\newblock Applications of random sampling in computational geometry, {II}.
\newblock {\em Discrete Comput. Geom.}, 4:387--421, 1989.

\bibitem[dBHTT10]{bhst-sqgqi-10}
M.~de~Berg, H.~Haverkort, S.~Thite, and L.~Toma.
\newblock Star-quadtrees and guard-quadtrees: {I/O}-efficient indexes for fat
  triangulations and low-density planar subdivisions.
\newblock {\em Comput. Geom. Theory Appl.}, 43:493--513, July 2010.

\bibitem[DHS01]{dhs-pc-01}
R.~O. Duda, P.~E. Hart, and D.~G. Stork.
\newblock {\em Pattern Classification}.
\newblock Wiley-Interscience, New York, 2nd edition, 2001.

\bibitem[DW82]{dw-nnmd-82}
L.~Devroye and T.J. Wagner.
\newblock Handbook of statistics.
\newblock In P.~R. Krishnaiah and L.~N. Kanal, editors, {\em Nearest neighbor
  methods in discrimination}, volume~2. North-Holland, 1982.

\bibitem[FH49]{fh-dandc-49}
E.~Fix and J.~Hodges.
\newblock Discriminatory analysis. nonparametric discrimination: Consistency
  properties.
\newblock {\em Technical Report 4, Project Number 21-49-004, USAF School of
  Aviation Medicine, Randolph Field, TX}, 1949.

\bibitem[GG91]{gg-vqsc-91}
A.~Gersho and R.~M. Gray.
\newblock {\em Vector Quantization and Signal Compression}.
\newblock Kluwer Academic Publishers, 1991.

\bibitem[GMM11]{gmm-wkd-11}
L.~J. Guibas, Q.~M{\'e}rigot, and D.~Morozov.
\newblock Witnessed $k$-distance.
\newblock In {\em Proc. 27th Annu. ACM Sympos. Comput. Geom.}, pages 57--64,
  2011.

\bibitem[{Har}99]{h-caspm-99}
S.~{Har-Peled}.
\newblock {Constructing approximate shortest path maps in three dimensions}.
\newblock {\em SIAM J. Comput.}, 28(4):1182--1197, 1999.

\bibitem[{Har}01]{h-rvdnl-01}
S.~{Har-Peled}.
\newblock A replacement for {Voronoi} diagrams of near linear size.
\newblock In {\em Proc. 42nd Annu. IEEE Sympos. Found. Comput. Sci.}, pages
  94--103, 2001.

\bibitem[{Har}06]{h-cdic-06}
S.~{Har-Peled}.
\newblock Coresets for discrete integration and clustering.
\newblock In {\em Proc. 26th Conf. Found. Soft. Tech. Theoret. Comput. Sci.},
  pages 33--44, 2006.

\bibitem[{Har}11]{h-gaa-11}
S.~{Har-Peled}.
\newblock {\em Geometric Approximation Algorithms}.
\newblock Amer. Math. Soc., 2011.

\bibitem[HIM12]{him-anntr-12}
S.~{Har-Peled}, P.~Indyk, and R.~Motwani.
\newblock Approximate nearest neighbors: {Towards} removing the curse of
  dimensionality.
\newblock {\em Theory Comput.}, 8:321--350, 2012.
\newblock Special issue in honor of Rajeev Motwani.

\bibitem[HK11]{hk-annsl-11}
S.~{Har-Peled} and N.~Kumar.
\newblock Approximate nearest neighbor search for low dimensional queries.
\newblock In {\em Proc. 22nd ACM-SIAM Sympos. Discrete Algs.}, pages 854--867,
  2011.

\bibitem[HK12]{hk-drhrp-12}
S.~{Har-Peled} and N.~Kumar.
\newblock Down the rabbit hole: Robust proximity search in sublinear space.
\newblock In {\em Proc. 53rd Annu. IEEE Sympos. Found. Comput. Sci.}, pages
  430--439, 2012.

\bibitem[HS11]{hs-rag-11}
S.~{Har-Peled} and M.~Sharir.
\newblock Relative $(p,\varepsilon)$-approximations in geometry.
\newblock {\em Discrete Comput. Geom.}, 45(3):462--496, 2011.

\bibitem[IM98]{im-anntr-98}
P.~Indyk and R.~Motwani.
\newblock Approximate nearest neighbors: {Towards} removing the curse of
  dimensionality.
\newblock In {\em Proc. 30th Annu. ACM Sympos. Theory Comput.}, pages 604--613,
  1998.

\bibitem[Koh01]{k-som-01}
T.~Kohonen.
\newblock {\em Self-Organizing Maps}, volume~30 of {\em Springer Series in
  Information Sciences}.
\newblock Springer, Berlin, 2001.

\bibitem[KOR00]{kor-esann-00}
E.~Kushilevitz, R.~Ostrovsky, and Y.~Rabani.
\newblock Efficient search for approximate nearest neighbor in high dimensional
  spaces.
\newblock {\em SIAM J. Comput.}, 2(30):457--474, 2000.

\bibitem[LLS01]{lls-ibscl-01}
Y.~Li, P.~M. Long, and A.~Srinivasan.
\newblock Improved bounds on the sample complexity of learning.
\newblock {\em J. Comput. Syst. Sci.}, 62(3):516--527, 2001.

\bibitem[Mat92]{m-ept-92}
J.~Matou{\v s}ek.
\newblock Efficient partition trees.
\newblock {\em Discrete Comput. Geom.}, 8:315--334, 1992.

\bibitem[Mat02]{m-ldg-02}
J.~Matou{\v s}ek.
\newblock {\em Lectures on Discrete Geometry}.
\newblock Springer, 2002.

\bibitem[Mei93]{m-plah-93}
S.~Meiser.
\newblock Point location in arrangements of hyperplanes.
\newblock {\em Inform. Comput.}, 106:286--303, 1993.

\bibitem[M{\' e}r13]{m-lbfkd-13}
Q.~M{\' e}rigot.
\newblock Lower bounds for $k$-distance approximation.
\newblock In {\em Proc. 29th Annu. ACM Sympos. Comput. Geom.}, page to appear,
  2013.

\bibitem[MP69]{mp-p-69}
M.~Minsky and S.~Papert.
\newblock {\em Perceptrons}.
\newblock MIT Press, Cambridge, MA, 1969.

\bibitem[MS94]{ms-trn-94}
T.~Martinetz and K.~Schulten.
\newblock Topology representing networks.
\newblock {\em Neural Netw.}, 7(3):507--522, March 1994.

\bibitem[Rup95]{r-draqt-95}
J.~Ruppert.
\newblock A {Delaunay} refinement algorithm for quality 2-dimensional mesh
  generation.
\newblock {\em J. Algorithms}, 18(3):548--585, 1995.

\bibitem[SDI06]{sdi-nnmlv-06}
G.~Shakhnarovich, T.~Darrell, and P.~Indyk.
\newblock {\em Nearest-Neighbor Methods in Learning and Vision: Theory and
  Practice (Neural Information Processing)}.
\newblock The MIT Press, 2006.

\bibitem[Sil86]{s-desda-86}
B.W. Silverman.
\newblock {\em Density estimation for statistics and data analysis}.
\newblock Monographs on statistics and applied probability. Chapman and Hall,
  1986.

\bibitem[SWY75]{swy-vsmai-75}
G.~Salton, A.~Wong, and C.~S. Yang.
\newblock A vector space model for automatic indexing.
\newblock {\em Commun. ACM}, 18:613--620, 1975.

\bibitem[Ten98]{t-mmpo-98}
J.~B. Tenenbaum.
\newblock Mapping a manifold of perceptual observations.
\newblock {\em Adv. Neur. Inf. Proc. Sys. 10}, pages 682--688, 1998.

\bibitem[WS06]{ws-ulms-04}
K.~Q. Weinberger and L.~K. Saul.
\newblock Unsupervised learning of image manifolds by semidefinite programming.
\newblock {\em Int. J. Comput. Vision}, 70(1):77--90, October 2006.

\end{thebibliography}

\end{document}